\documentclass[12pt]{article}
\usepackage{amssymb, amsmath, amsthm}
\usepackage{lscape} 
\usepackage{mathdots}
\usepackage{thmtools}
\usepackage{mathtools}
\declaretheorem[style=definition]{example}
\usepackage{appendix}
\usepackage{bm}
\usepackage{bbm}
\usepackage{caption}
\usepackage{enumerate} 
\usepackage{graphicx}
\usepackage{natbib}
\usepackage{setspace}
\usepackage{hyperref} 
\usepackage{color}
\usepackage{tikz}
\usepackage{mathrsfs}
\usetikzlibrary{snakes}
\usepackage{subfigure}
\numberwithin{equation}{section}
\usepackage[top=1.25in, bottom=1.25in, left=1.25in, right=1.25in]{geometry}

\DeclareMathOperator{\sgn}{sgn}
\newtheorem{assm}{Assumption}
\newtheorem{lem}{Lemma}
\newtheorem{prop}{Proposition}
\newtheorem{thm}{Theorem}
\newtheorem{cor}{Corollary}

\newtheorem{defn}{Definition}

\begin{document}
\renewcommand{\thefootnote}{\fnsymbol{footnote}}
\renewcommand\thmcontinues[1]{Continued}
\title{Optimal Incentive Contract with Endogenous Monitoring Technology}
\author{Anqi Li\footnote{Department of Economics, Washington University in St. Louis. anqili@wustl.edu.}
\and Ming Yang\footnote{Fuqua School of Business, Duke University. ming.yang@duke.edu. We thank the coeditor, four anonymous referees, Nick Bloom, George Mailath, Ilya Segal, Chris Shannon, Joel Sobel, Jacob Steinhardt, Bob Wilson and the seminar participants at Caltech, Decentralization Conference 2017, Duke-UNC, Johns Hopkins Carey, Northwestern, Stanford, UCSB, UCSD, U of Chicago and U of Washington for comments and suggestions. Lin Hu and Jessie Li provided generous assistance for the numerical analysis. All errors are our own.}}
\date{\emph{Forthcoming, Theoretical Economics}}
\maketitle

\begin{abstract}
\noindent Recent technology advances have enabled firms to flexibly process and analyze sophisticated employee performance data at a reduced and yet significant cost. We develop a theory of optimal incentive contracting where the monitoring technology that governs the above procedure is part of the designer's strategic planning. In otherwise standard principal-agent models with moral hazard, we allow the principal to partition agents' performance data into any finite categories and to pay for the amount of information the output signal carries. Through analysis of the trade-off between giving incentives to agents and saving the monitoring cost, we obtain characterizations of optimal monitoring technologies such as information aggregation, strict MLRP, likelihood ratio-convex performance classification, group evaluation in response to rising  monitoring costs, and assessing multiple task performances according to agents' endogenous tendencies to shirk. We examine the implications of these results for workforce management and firms' internal organizations.

\bigskip


\bigskip


\end{abstract} 
\renewcommand{\thefootnote}{\arabic{footnote}}
\pagebreak

\section{Introduction}\label{sec_introduction}
Recent technology advances have enabled firms to flexibly process and analyze sophisticated employee performance data at a reduced and yet significant cost.  Speech analytics software, natural language processing tools and cloud-based systems are increasingly used to convert hard-to-process contents into succinct and meaningful ratings such as ``satisfactory'' and ``unsatisfactory'' (\cite{nlp}; \cite{nytimes}; \cite{harper}). This paper develops a theory of optimal incentive contracting where the \emph{monitoring technology} that governs the above procedure is part of the designer's strategic planning. 

Our research agenda is motivated by the case of call center performance management reported by \cite{nytimes}. It has long been recognized that the conversations between call center agents and customers contain useful performance indicators such as customer sentiment, voice quality and tone, etc.. Recently, the advent of speech analytics software has finally enabled the processing and analysis of these contents, as well as their conversions into meaningful ratings such as ``satisfactory'' and  ``unsatisfactory.'' On the one hand, running speech analytics software consumes server space and power, and the procedure has been increasingly outsourced to third parties to take advantage of the latest development in cloud computing. On the other hand, managers now have considerable freedom to decide which facets of the customer conversation to utilize, thanks to the increased availability of products whose specialties range from emotion detection to word spotting. 

We formalize the flexibility and cost associated with the design and implementation of the monitoring technology in otherwise standard principal-agent models with moral hazard. Specifically, we allow the monitoring technology to partition agents' performance data into any finite categories, at a cost that increases with the amount of information the output signal carries  (hereafter \emph{monitoring cost}). An incentive contract pairs the monitoring technology with a wage scheme that maps realizations of the output signal to different wages. An optimal contract minimizes the sum of expected wage and monitoring cost, subject to agents' incentive constraints. 

Our main result gives characterizations of optimal monitoring technologies in general environments, showing that the assignment of Lagrange multiplier-weighted likelihood ratios to performance categories is positive assortative in the direction of agent utilities. Geometrically, this means that optimal monitoring technologies comprise convex cells in the space of likelihood ratios or their transformations. This result provides practitioners with the needed formula for sorting employee performance data, and exploiting its geometry yields insights into workforce management and firms' internal organizations. 

Our proof strategy works directly with the principal's Lagrangian. It handles general situations featuring multiple agents and tasks, in which the direction of sorting vector-valued likelihood ratios is nonobvious a priori. It overcomes the technical challenge whereby perturbations of the sorting algorithm affect wages endogenously through the Lagrange multipliers of agents' incentive constraints, yielding effects that are new and difficult to assess using standard methods.    

We give three applications of our result. In the single-agent model considered in \cite{holmstrom}, we show that the assignment of likelihood ratios to wage categories is positive assortative and follows a simple cutoff rule. The monitoring technology aggregates potentially high-dimensional performance data into rank-ordered ratings, and the output signal satisfies the strict monotone likelihood ratio property with respect to the order induced by likelihood ratios. Solving cutoff likelihood ratios yields consistent findings with recent developments in manufacturing, retail and healthcare sectors, where decreases in the data processing cost have shown to increase the fineness of the performance grids (\cite{bloomvanreenenwp}; \cite{nlp}; \cite{mckinsey}). 

In the multi-agent model considered in \cite{holmstromteam}, the optimal monitoring technology partitions vectors of individual agents' likelihood ratios into convex polygons. Based on this result, we then compare individual and group performance evaluations from the angle of monitoring cost, showing that firms should switch from individual to group evaluation in response to rising monitoring costs. This result formalizes the theses of \cite{ad} and \cite{lazearrosen} that either team or tournament should be the dominant incentive system when individual performance evaluation is too costly to conduct. It is consistent with the findings of \cite{bloomvanreenenwp} that lack of IT access increases the use of group performance evaluation among otherwise similar firms. 

In the multiple-task model studied in \cite{holmstrommilgrom}, the resources spent on the assessment of a task performance should increase with the agent's endogenous tendency to shirk the corresponding task. Using simulation, we apply this result to the study of, e.g., 
how improved precision of some task measurements (caused by, e.g., the advent of high-quality scanner data measuring the skillfulness in scanning items) would affect the resources spent on the assessments of other task performances (e.g., projecting warmth to customers). 

\subsection{Related Literature}\label{sec_literature}
Earlier studies on contracting with costly experiments (in the sense of \cite{blackwell}) include, but are not limited to: \cite{baimandemski} and \cite{dye2}, in which the principal can pay an external auditor for drawing a signal from an exogenous distribution; \cite{holmstrom}, \cite{grossmanhart} and \cite{kim}, in which signal distributions are ranked based on the incentive costs they incur. In these studies, the principal can change the probability space generated by the agent's hidden effort and, in the first two studies, through paying stylized costs. In contrast, we focus on the conversion of raw data into performance ratings while taking the probability space as given. Also our assumption that the monitoring cost increases with the amount of information carried by the output signal could be ill-suited for modeling the cost of running experiments.

The current work differs from existing studies on rational inattention (hereafter RI) in three aspects. First, early developments in RI by \cite{sims}, \cite{riprice} and \cite{woodford} sought to explain the stickiness of macroeconomic variables by information processing costs, whereas we examine the implication of costly yet flexible monitoring for principal-agent relationships.\footnote{\cite{yang} studies a security design problem where a rationally inattentive buyer can obtain any signal about the uncertain fundamental at a cost that is proportional to entropy reduction. Other recent efforts to introduce RI into strategic environments include but are not limited to: \cite{riconsumer}, \cite{martin} and \cite{ravid}.} Second, we focus mainly on partitional monitoring technologies because in reality, adding non-performance-related factors into employee ratings could have dire consequences such as appeals, lawsuits and excessive turnover.\footnote{See standard HR textbooks for this subject matter. \cite{saintpaul} demonstrates the validity of entropy as an information cost in  decision problems where the decision variable is a deterministic function of the exogenous state variable. } Finally, our monitoring cost function nests entropy as a special case.

Recent works of \cite{cremeretal}, \cite{voronoi}, \cite{sobel} and \cite{dilme} examine the optimal language used between organization members who share a common interest but face communication costs. The absence of conflicting interests hence incentive constraints distinguishes these works from ours.

\bigskip 

The remainder of this paper is organized as follows: Section \ref{sec_baseline} introduces the baseline model; Section \ref{sec_analysis} presents main results; Sections \ref{sec_multiagent} and \ref{sec_multidev} investigate extensions of the baseline model; Section \ref{sec_conclusion} concludes. See Appendices \ref{sec_proof} and \ref{sec_online} for omitted proofs and additional materials. 

\section{Baseline Model}\label{sec_baseline}
\subsection{Setup}\label{sec_setup}
\paragraph{Primitives} A risk-neutral principal faces a risk-averse agent, who  earns a utility $u\left(w\right)$ from spending a nonnegative wage $w \geq 0$ and incurs a cost $c\left(a\right)$ from privately exerting high or low effort $a \in \{0,1\}$. The function $u: \mathbb{R}_+ \rightarrow \mathbb{R}$ satisfies $u\left(0\right)=0$, $u'>0$ and $u''<0$, whereas $c\left(1\right)=c>c\left(0\right)=0$. 

Each effort choice $a$ generates a probability space $\left(\Omega, \Sigma, P_a\right)$, where $\Omega$ is a finite-dimensional Euclidean space that comprises the agent's performance data, $\Sigma$ is the Borel sigma-algebra on $\Omega$, and $P_a$ is the probability measure on $\left(\Omega, \Sigma\right)$. $P_a$'s are assumed to be mutually absolutely continuous, and the probability density function $p_a$'s they induce are well-defined and everywhere positive. 

\paragraph{Incentive contract}  An incentive contract $\langle \mathcal{P}, w\left(\cdot\right) \rangle$ is a pair of \emph{monitoring technology} $\mathcal{P}$ and \emph{wage scheme} $w: \mathcal{P} \rightarrow \mathbb{R}_{+}$. The former represents a human- or machine-operated system that governs the processing and analysis of performance data, whereas the latter maps outputs of the first-step procedure to different levels of wages. In the main body of this paper, $\mathcal{P}$ can be any partition of $\Omega$ with at most $K$ cells that are all of positive measures,\footnote{In Appendix \ref{sec_random}, we allow the monitoring technology to be any mapping from $\Omega$ to lotteries on finite performance categories. If the lottery is degenerate, then the monitoring technology is partitional.} and $w: \mathcal{P} \rightarrow \mathbb{R}_{+}$ maps each cell $A$ of $\mathcal{P}$ to a nonnegative wage $w\left(A\right) \geq 0$.\footnote{Appendix \ref{sec_ir} examines the case where the agent is constrained by individual rationality.  } The upper bound $K$ for the \emph{rating scale} $|\mathcal{P}|$ can be any integer greater than one and will be taken as given throughout the analysis.\footnote{The upper bound $K$, while stylized, guarantees the existence of optimal incentive contract(s). Judging from the simulation exercises we have so far conducted, the optimal rating scale is typically smaller than $K$ even when $\mu$ is small (see, e.g., Figure \ref{figure_ratingscale}).}

For any data point $\omega \in \Omega$, let $A\left(\omega\right)$ be the unique \emph{performance category} that contains $\omega$ and let $w\left(A\left(\omega\right)\right)$ be the wage associated with $A\left(\omega\right)$. Time evolves as follows: 
\begin{enumerate}
\item the principal commits to $\langle \mathcal{P}, w\left(\cdot\right)\rangle$; 
\item the agent privately chooses $a \in \left\{0,1\right\}$;
\item Nature draws $\omega$ from $\Omega$ according to $P_a$; 
\item the monitoring technology outputs $A\left(\omega\right)$;  
\item the principal pays $w\left(A\left(\omega \right)\right)$ to the agent.
\end{enumerate}

\paragraph{Implementation cost} For any given effort choice $a$ by the agent, a monitoring technology $\mathcal{P}=\left\{A_1,\cdots, A_N\right\}$ outputs a signal $X: \Omega \rightarrow \mathcal{P}$ whose probability distribution is represented by a vector $\bm{\pi}\left(\mathcal{P},a\right)=\left(P_a\left(A_1\right),\cdots, P_a\left(A_N\right), 0,\cdots, 0\right)$ in the $K$-dimensional simplex. The principal incurs the following cost from implementing an incentive contract $\langle \mathcal{P}, w\left(\cdot\right) \rangle$: 
\[
\sum_{A \in \mathcal{P}} P_a
\left(A
\right)w\left(A\right)+ \mu \cdot H\left(\mathcal{P}, a\right), 
\]
which consists of two parts. The first part $\sum_{A \in \mathcal{P}} P_a\left(A\right)w\left(A\right)$, i.e., the \emph{incentive cost}, has been the central focus of the existing principal-agent literature. The second part $\mu \cdot H\left(\mathcal{P}, a\right)$, hereafter termed the \emph{monitoring cost},  represents the cost associated with the processing and analysis of the performance data. In particular, $\mu>0$ is an exogenous parameter which we will further discuss in Section \ref{sec_implication}, whereas $H\left(\mathcal{P},a\right)$ captures the amount of information carried by the output signal and is assumed to satisfy the following properties:


\begin{assm}\label{assm_mc}
There exists a function $h: \Delta^K \rightarrow \mathbb{R}_+$ such that $H\left(\mathcal{P}, a\right)=h\left(\bm{\pi}\left(\mathcal{P},a\right)\right)$ for all  $\left(\mathcal{P}, a\right)$. Furthermore,
\begin{enumerate}[(a)]
\item $h\left(\pi_1, \cdots, \pi_K\right)=h\left(\pi_{\sigma\left(1\right)}, \cdots, \pi_{\sigma\left(K\right)}\right)$ for all probability vector $\left(\pi_1,\cdots, \pi_K\right) \in \Delta^K$ and permutation $\sigma$ on $\left\{1,\cdots, K\right\}$; 
\item $h\left(0, \pi_2, \cdots\right)<h\left(\pi_1', \pi_2'', \cdots\right)$ for all $\left(0, \pi_2, \cdots\right)$ and $\left(\pi_1', \pi_2', \cdots\right) \in \Delta^K$ that differ only in the first two elements and satisfy $\pi_2, \pi_1', \pi_2'>0$ and $\pi_2=\pi_1'+\pi_2'$. 
\end{enumerate}
\end{assm}

Inspired by basic principles of information theory, Assumption \ref{assm_mc} stipulates that the amount of information carried by the output signal should depend only on the latter's probability distribution and must increase with the fineness of the monitoring technology. Aside from probabilities, nothing else matters, not even the naming or contents of the performance categories. Assumption \ref{assm_mc} is satisfied by, e.g., the entropy $- \sum_{A \in \mathcal{P}} P_a\left(A\right) \log_2 P_a\left(A\right)$ of the output signal and the bits of information $\log_2|\mathcal{P}|$ it carries.\footnote{The bit is a basic unit of information in information theory, computing, and digital communications. In information theory, one bit is defined as the maximum information entropy of a binary random variable.} In Section \ref{sec_cost}, we motivate the use of this assumption in the example of call center performance management. 

\paragraph{The principal's problem}  Consider the problem of inducing high effort from the agent.\footnote{The problem of inducing low effort is standard.} Define a random variable $Z: \Omega \rightarrow \mathbb{R}$ by 
\[
Z\left(\omega\right)=1-\frac{p_0\left(\omega\right)}{p_1\left(\omega\right)} \text{ } \text{ }\forall \omega,
\]
where $p_0\left(\omega\right)/p_1\left(\omega\right)$ is the \emph{likelihood ratio} associated with data point $\omega$. Note that $\mathbf{E}\left[Z \mid a=1\right]=0$ and that the range of $Z$ is a subset of $\left(-\infty, -1\right)$. For any set $A \in \Sigma$ of positive measure, define the \emph{$z$-value} of $A$ by
\[
z(A)=\mathbf{E}\left[Z \mid A; a=1\right].
\]
In words, $z\left(A\right)$ represents the average value of $Z$ conditional on the data point being drawn from $A$. 

A contract $\langle \mathcal{P}, w\left(\cdot\right)\rangle$ is incentive compatible if 
\[
\sum_{A \in \mathcal{P}} P_1\left(A\right)u\left(w\left(A\right)\right)- c \geq \sum_{A \in \mathcal{P}} P_0\left(A\right)u\left(w\left(A\right)\right)
\]
or, equivalently, 
\begin{equation}
\tag{IC} \sum_{A \in \mathcal{P}} P_1\left(A\right)u\left(w\left(A\right)\right)z\left(A\right) \geq c, \label{eqn_ic}
\end{equation}
and it satisfies the limited liability constraint if
\begin{equation}
\tag{LL}  w\left(A\right) \geq 0 \text{ } \forall A \in \mathcal{P}. \label{eqn_ll}
\end{equation}
An optimal incentive contract that induces high effort from the agent (optimal incentive contract for short) minimizes the total implementation cost under high effort, subject to the incentive compatibility constraint and limited liability constraint: 
\[
\min_{\langle \mathcal{P}, w\left(\cdot\right)\rangle} \sum_{A \in \mathcal{P}} P_1\left(A\right)w\left(A\right)+\mu \cdot H\left(\mathcal{P}, 1\right) \text{ s.t. (\ref{eqn_ic}) and (\ref{eqn_ll})}. 
\]
In what follows, we will denote the solution(s) to the above problem by $\langle \mathcal{P}^*, w^*\left(\cdot\right)\rangle$. 

\subsection{Monitoring Cost}\label{sec_cost}
We first illustrate Assumption \ref{assm_mc} in the context of call center performance management: 

\begin{example}[label=exa:cont1]
In the example described in Section \ref{sec_introduction}, a piece of performance data comprises the major characteristics of a call history (e.g., customer sentiment and voice quality) encoded in binary digits, and the  monitoring technology represents the speech analytics program that categorizes binary digits into performance ratings. To formalize the design flexibility, we allow the monitoring technology to partition performance data into any $N \leq K$ categories, where $K$ can be any interger greater than one. The cost of  running the monitoring technology is assumed to increase with the amount of processed information, whose definition varies from case to case. For example, if the monitoring technology runs many times among many identical agents, then the optimal design should minimize the average steps it takes to find the performance category containing the raw data point. By now, it is well known that this quantity equals approximately the entropy of the output signal. In contrast, if the monitoring technology runs only a few times for a few number of agents, then the worst-case (or unamortized) amount of processed information is best captured by the bits of information carried by the output signal (see, e.g., \cite{infortheory}). In both cases, the quantity of our interest  depends only on the probability distribution of the output signal and nothing else. 
\end{example}

We next introduce the concept of \emph{setup cost} and distinguish it from our  notion of monitoring cost:

\begin{example}[continues=exa:cont1]
As its name suggests, \emph{setup cost} refers the cost incurred to set up the infrastructure that facilitates data processing and analysis, e.g., Fast Fourier Transformation (FFT) chips (which transform sound waves into their major characteristics coded in binary digits), recording devices, etc..

The major role of setup cost is to change the probability space $\left(\Omega, \Sigma, P_a\right)$. For example, design improvements in FFT chips enable more frequent sampling of sound waves and cause $\left(\Omega, \Sigma, P_a\right)$ to change. In what follows, we will take the probability space as given and ignore the setup cost. That said, one can certainly embed our analysis into a two-stage setting in which the principal first incurs the setup cost and then the monitoring cost. Results below will carry over to this new setting.
\end{example}

\section{Analysis}\label{sec_analysis}
\subsection{Preview}\label{sec_preview}
\begin{example}\label{exm_main}
Suppose $u\left(w\right)=\sqrt{w}$, $Z$ is uniformly distributed over $\left[-1/2, 1/2\right]$ under $a=1$ and $H\left(\mathcal{P},a\right)=f\left(|\mathcal{P}|\right)$ for some strictly increasing function $f: \left\{1,\cdots, K \right\} \rightarrow \mathbb{R}_+$. Below we walk through the key steps in  solving the optimal incentive contract, give closed-form solutions and discuss their practical implications. 

\paragraph{Optimal wage scheme} We first solve for the optimal wage scheme for any given monitoring technology $\mathcal{P}$ as in \cite{holmstrom}. Specifically, label the performance categories as $A_1, \cdots, A_N$, and write $\pi_n=P_1\left(A_n\right)$ and $z_n=z\left(A_n\right)$ for $n=1,\cdots, N$. Assume $z_j \neq z_k$ for some $j, k \in \left\{1,\cdots, N\right\}$ to make the analysis interesting.  The principal's problem is then
\begin{align*}
\min_{\left\{w_n\right\}} & \sum_{n=1}^N \pi_n w_n,\\
\tag{IC} \text{ s.t. } & \sum_{n=1}^N \pi_n \sqrt{w_n} z_n \geq c,\\
\tag{LL}\text{ and } & w_n \geq 0, n=1,\cdots, N. 
\end{align*}
Straightforward algebra yields the expression for minimal incentive cost:  
\begin{equation*}
c^2\left(\sum_{n=1}^N \pi_n \max\left\{0, z_n\right\}^2\right)^{-1}. \label{eqn_incentive_cost}
\end{equation*}
A close inspection reveals Holmstr\"{o}m's (1979) \emph{sufficient statistics principle}, namely $z$-value is the only part of the performance data that provides the agent with incentives.

\paragraph{Optimal monitoring technology} We next solve for the optimal monitoring technology. First, note that the principal should partition performance data based only on their $z$-values, and that different  performance categories must attain different $z$-values and wages. The reason combines the sufficient statistic principle with  Assumption \ref{assm_mc}(b), namely merging performance categories of the same $z$-value saves the monitoring cost while leaving the incentive cost unaffected and thus constitutes an improvement to the original monitoring technology. 

A more interesting question concerns how we should assign the various data points, identified by their $z$-values, to different performance categories. In the baseline model featuring a single agent and binary efforts, the answer to this question is relatively straightforward: assign high (resp. low) $z$-values to high-wage (resp. low-wage) categories. Here is a quick proof of this result: since the left-hand side of the (\ref{eqn_ic}) constraint is supermodular in wages and $z$-values, if our conjecture were false, then reshuffling data points as above while holding the probabilities of performance categories constant reduces the incentive cost while leaving the monitoring cost unaffected.

When extending the above intuition to general settings featuring multiple agents or multiple actions, we face two challenges. First, in the case where $z$-values and wages are vectors, the direction of sorting these objects is nonobvious a priori. Second, changes in the sorting algorithm affect wages endogenously through the Lagrange multipliers of the incentive constraints, yielding effects that are new and difficult to assess using standard methods. 

The proof strategy presented in Section \ref{sec_proofsketch} overcomes these challenges, showing that the assignment of Lagrange multiplier-weighted $z$-values to performance categories must be positive assortative in the direction of agent utilities. Geometrically, this means that any optimal monitoring technology must comprise convex cells in the space of $z$-values or their transformations. Theorems \ref{thm_main}, \ref{thm_multiagent} and \ref{thm_multidev} formalize the above statements. 

\paragraph{Implications}  An important feature of the optimal monitoring technology is \emph{information aggregation}---a term used by human resource practitioners to refer to the aggregation of potentially high-dimensional performance data into rank-ordered ratings such as ``satisfactory'' and ``unsatisfactory.'' 

The geometry of the optimal monitoring technology sheds light on the practical issues covered in Sections \ref{sec_implication}, \ref{sec_multiagent_application} and \ref{sec_multitask_application}. Consider, for example, optimal performance grids. In the current example, it can be shown that the optimal $N$-partitional monitoring technology divides the space $\left[-1/2, 1/2\right]$ of $z$-values into $N$ disjoint intervals $[\widehat{z}_{n-1}, \widehat{z}_n)$, $n=1,\cdots, N$, where $\widehat{z}_0=-1/2$ and $\widehat{z}_N=1/2$. The optimal cut points $\left\{\widehat{z}_n\right\}_{n=1}^{N-1}$ can be solved as follows: 
\[\min_{\left\{\widehat{z}_n\right\}} c^2\left(\sum_{n=1}^N \pi_{n} \max\{0, z_{n}\}^2\right)^{-1} - \mu \cdot f\left(N\right),\]
where 
\[\pi_{n}=\int_{\widehat{z}_{n-1}}^{\widehat{z}_{n}}  dZ =\widehat{z}_{n}-\widehat{z}_{n-1},\]
and 
\[
z_{n}=\frac{1}{\pi_{n}}\int_{\widehat{z}_{n-1}}^{\widehat{z}_{n}} Z dZ=\frac{1}{2}\left(\widehat{z}_{n-1}+\widehat{z}_{n}\right).\]
Straightforward algebra yields
\[\widehat{z}_{n}=\frac{2n-1}{4N-2}, \text{ } n=1,\cdots, N-1.\]
Based on this result, as well as the functional form of $f$, we can then solve for  the optimal rating scale $N$ and hence the optimal incentive contract completely. 
\end{example}

\subsection{Main Results}\label{sec_main}
This section analyzes optimal incentive contracts. Results below hold true except perhaps on a measure zero set of data points. The same disclaimer applies to the remainder of this paper.

We first define \emph{$Z$-convexity}: 

\begin{defn}\label{defn_zconvex}
A set $A \in \Sigma$ is \emph{$Z$-convex} if the following holds for all $\omega',\omega'' \in A$ such that $Z\left(\omega'\right) \neq Z\left(\omega''\right)$:   
\[
\left\{\omega \in \Omega: Z\left(\omega\right)=\left(1-s\right) \cdot Z\left(\omega'\right)+s \cdot Z\left(\omega''\right) \text{ for some } s \in \left(0,1\right)\right\} \subset A.
\]
\end{defn}

In words, a set $A \in \Sigma$ is $Z$-convex if whenever it contains data points of different $z$-values, it must also contain all data points of intermediate $z$-values. Let $Z\left(A\right)$ denote the image of any set $A \in \Sigma$ under mapping $Z$. In the case where $Z\left(\Omega\right)$ is a connected set in $\mathbb{R}$, the above definition is equivalent to the convexity of $Z\left(A\right)$ in $\mathbb{R}$. 

A few assumptions before we go into detail. The next assumption says that the distribution of $Z$ has no atom or hole:

\begin{assm}\label{assm_regular}
$Z$ is distributed atomlessly on a connected set $Z\left(\Omega\right)$ in $\mathbb{R}$ under $a=1$. 
\end{assm}

The next assumption says that $Z\left(\Omega\right)$ is compact set in $\mathbb{R}$: 
\begin{assm}\label{assm_compact}
$Z\left(\Omega\right)$ is a compact set in $\mathbb{R}$.
\end{assm}

The next assumption imposes regularities on the monitoring cost function: Part (a) of it holds for the bits of information carried by the output signal, and Part (b) of it holds for the entropy of the output signal: 

\begin{assm}\label{assm_cf}
The function $h: \Delta^K \rightarrow \mathbb{R}_+$ satisfies one of the following conditions: 
\begin{enumerate}[(a)]
\item $h\left(\bm{\pi}\left(\mathcal{P},a\right)\right)=f\left(|\mathcal{P}|\right)$ for some strictly increasing function $f: \left\{1,\cdots, K\right\} \rightarrow \mathbb{R}_+$;
\item $h$ is continuous. 
\end{enumerate}
\end{assm}

We now state our main results. The next theorem shows that any optimal incentive contract assigns data points of high (resp. low) $z$-values to high-wage (resp. low-wage) categories. Under Assumption \ref{assm_regular}, this can be achieved by first dividing $z$-values into disjoint intervals and then backing out the partition of the original data space accordingly. The result is an aggregation of potentially high-dimensional data into rank-ordered ratings, as well as a wage scheme that is strictly increasing in these ratings: 

\begin{thm}\label{thm_main}
Assume Assumption \ref{assm_mc} and let $\langle \mathcal{P}^*, w^*\left(\cdot\right)\rangle$ be any optimal incentive contract that induces high effort from the agent. Then $\mathcal{P}^*$ comprises $Z$-convex cells labeled as $A_1,\cdots, A_N$ where $0=w^*\left(A_1\right)<\cdots<w^*\left(A_N\right)$. Assume, in addition, Assumption \ref{assm_regular}. Then there exist $\inf Z\left(\Omega\right)= \widehat{z}_0< \widehat{z}_1 < \cdots< \widehat{z}_N = \sup Z(\Omega)$ such that $A_n =\left\{\omega: Z\left(\omega\right) \in \left[\widehat{z}_{n-1}, \widehat{z}_{n}\right)\right\}$ for $n=1,\cdots, N$.\footnote{Under Assumption \ref{assm_regular}, the set of (finite) cut points has measure zero, so it is unimportant which of the two adjacent intervals a cut point belongs to. The choice of expressing all intervals as right half-open ones is purely aesthetic.} 
\end{thm} 

The next theorem proves existence of optimal incentive contract: 

\begin{thm}\label{thm_exist}
An optimal incentive contract that induces high effort from the agent exists under Assumptions \ref{assm_mc}-\ref{assm_cf}. 
\end{thm}

\begin{proof}
See Appendix \ref{sec_proof_baseline}.
\end{proof}

\subsection{Proof Sketch}\label{sec_proofsketch}
The proof of Theorem \ref{thm_main} consists of three steps. The intuitions of steps one and two have already been discussed in Example \ref{exm_main}. Step three is new.

\paragraph{Step one} We first take any monitoring technology $\mathcal{P}$ as given and solve for the optimal wage scheme as in \cite{holmstrom}: 
\begin{equation}
\min_{w: \mathcal{P} \rightarrow \mathbb{R}_{+}}\sum_{A \in \mathcal{P}} P_1\left(A\right)w\left(A\right) \text{ s.t.  (\ref{eqn_ic}) and (\ref{eqn_ll}).} \label{eqn_reducedproblem}
\end{equation}
The next lemma restates Holmstr\"{o}m's (1979) \emph{sufficient statistic principle}:
\begin{lem}\label{lem_foc}
Let $w^*\left(\cdot; \mathcal{P}\right)$ be any solution to Problem  (\ref{eqn_reducedproblem}). Then there exists $\lambda>0$ such that $u'\left(w^*\left(A; \mathcal{P}\right)\right)=1/\left(\lambda z\left(A\right)\right)$ for all $A \in \mathcal{P}$ such that $w^*\left(A; \mathcal{P}\right)>0$. 
\end{lem}

\begin{proof}
See Appendix \ref{sec_proof_baseline}. 
\end{proof}

\paragraph{Step two} We next demonstrate that different performance categories must attain different $z$-values and wages:  

\begin{lem}\label{lem_mlrp}
Assume Assumption \ref{assm_mc}. Let $\langle \mathcal{P}^*, w^*\left(\cdot\right) \rangle$ be any optimal incentive contract that induces high effort from the agent and label the cells of $\mathcal{P}^*$ as $A_1,\cdots, A_N$ such that $z\left(A_1\right)\leq \cdots \leq z\left(A_N\right)$. Then $z\left(A_1\right)<0<\cdots<z\left(A_N\right)$ and $0=w^*\left(A_1\right)<\cdots<w^*\left(A_N\right)$. 
\end{lem}

\begin{proof}
See Appendix \ref{sec_proof_baseline}. 
\end{proof}

\paragraph{Step three} We finally demonstrate that the assignment of $z$-values into wage categories is positive assortative. In Example \ref{exm_main}, we sketched a proof based on supermodularity and pointed out the difficulties of extending that argument to multidimensional environments.  The argument below overcomes these difficulties.

Take any optimal incentive contract with distinct performance categories $A_j$ and $A_k$. From Lemma \ref{lem_mlrp}, we know that $z\left(A_j\right)\neq z\left(A_k\right)$. Fix any $\epsilon>0$, and take any $A'_\epsilon\subset A_j$ and $A''_{\epsilon}\subset A_k$ such that $P_1\left(A'_\epsilon\right)=P_1\left(A''_\epsilon\right)=\epsilon$ and $z\left(A'_\epsilon\right)=z' \neq z\left(A''_\epsilon\right)=z''$. In words,  $A'_\epsilon$ and $A''_\epsilon$ have the same probability $\epsilon$ under $a=1$ but different $z$-values that are independent of $\epsilon$. Lemma \ref{lem_smallSet} of Appendix \ref{sec_proof_baseline_lemma} proves existence of $A'_\epsilon$ and $A''_\epsilon$ when $\epsilon$ is small. 

Consider a perturbation to the monitoring technology that ``swaps'' $A'_\epsilon$ and $A''_\epsilon$. Post the perturbation, the new performance categories, denoted by $A_n\left(\epsilon\right)$'s, become $A_j\left(\epsilon\right)=\left(A_j \backslash A'_\epsilon\right)\cup A''_\epsilon$, $A_k\left(\epsilon\right)=\left(A_k \backslash A''_\epsilon\right)\cup A'_\epsilon$ and $A_n\left(\epsilon\right) =A_n$ for $n \neq j,k$. Since the perturbation has no effect on the probabilities of the performance categories under $a=1$, it does not affect the monitoring cost by Assumption \ref{assm_mc}(a). Meanwhile, it changes the principal's Lagrangian to the following (ignore the (\ref{eqn_ll}) constraint): 
\[
\mathcal{L}\left(\epsilon\right)=\sum_{n} \pi_n \left(w_n\left(\epsilon\right)-\lambda(\epsilon)u\left(w_n\left(\epsilon\right)\right)z_n \right)+\lambda\left(\epsilon\right)c, 
\]
where $\pi_n$ denotes the probability of $A_n$ (equivalently $A_n\left(\epsilon\right)$) under $a=1$, $w_n\left(\epsilon\right)$ the optimal wage at $A_n\left(\epsilon\right)$, and $\lambda\left(\epsilon\right)$ the Lagrange multiplier associated with the (\ref{eqn_ic}) constraint. A close inspection of the Lagrangian leads to the following conjecture: to minimize $\mathcal{L}\left(\epsilon\right)$, the assignment of Lagrange multiplier-weighted $z$-values to performance categories must be positive assortative in the direction of agent utilities. 

To develop intuition, we assume differentiability and obtain
\begin{align*}
\mathcal{L}'\left(0
\right)=&\sum_{n} \pi_n w'_n\left(0\right)-\lambda'\left(0\right)\underbrace{\left(\sum_{n} \pi_n u\left(w_n\left(0\right)\right) z_n\left(0\right)-c\right)}_\text{\quad \quad \text{ } $(1)$\quad$=0$}\\
&-\lambda(0)\left(\sum_{n} \pi_n  \cdot \underbrace{u'\left(w_n\left(0\right)\right) z_n\left(0\right)}_\text{\quad \quad \quad \quad \quad \text{ } $(2)$\quad $=1/\lambda(0)$} \cdot {w'_n}\left(0\right)+\sum_{n} \pi_n u\left(w_n\left(0\right)\right) z'_n\left(0\right)\right)\\
=&\sum_{n} \pi_n w'_n\left(0\right)-0-\sum_{n} \pi_n w'_n\left(0\right)-\lambda\left(0\right)\sum_{n} \pi_n u\left(w_n\left(0\right)\right) z'_n\left(0\right)\\
=&-\lambda\left(0\right)\sum_{n} \pi_n u\left(w_n\left(0\right)\right) z'_n\left(0\right)\\
= &\text{ } \lambda\left(0\right)\left(z''-z'\right)\left(u\left(w_k\left(0\right)\right)-u\left(w_j\left(0\right)\right)\right).
\end{align*}
In the above expression, $(1)=0$ because the (\ref{eqn_ic}) constraint binds under the original contract, and $(2)=1/\lambda\left(0\right)$ by Lemma \ref{lem_foc}. These findings resolve our concerns raised in Section \ref{sec_preview}, showing that the effects of our perturbation on the Lagrange multiplier and wages are negligible.

To complete the proof, note that $\mathcal{L}'\left(0\right) \geq 0$ by optimality, and that $\mathcal{L}'\left(0\right) \neq 0$ because $\lambda\left(0\right)>0$, $z'' \neq z'$ and $w_j\left(0\right)\neq w_k\left(0\right)$ (Lemma \ref{lem_mlrp}). Combining yields $\mathcal{L}'\left(0\right)>0$, so our conjecture is indeed true. $Z$-convexity is immediate: if a performance category contains extreme but not intermediate $z$-values, then the assignment of $z$-values goes in the wrong direction and an improvement can be constructed.

The above proof strategy yields the endogenous direction of sorting raw data into performance categories, which is relatively straightforward in the baseline model but is less so in later extensions. The proof in Appendix \ref{sec_proof_baseline} does not assume differentiability and handles the limited liability constraint, too. 

\subsection{Implications}\label{sec_implication}
\paragraph{Strict MLRP} Theorem \ref{thm_main} implies that the signal generated by any optimal monitoring technology must satisfy the strict monotone likelihood ratio property (hereafter \emph{strict MLRP}) with respect to the order induced by $z$-values: 

\begin{defn}
For any $A, A' \in \Sigma$ of positive measures, write $A \overset{z}{\preceq} A'$ if $z\left(A\right) \leq z\left(A'\right)$. 
\end{defn}

\begin{cor}\label{cor_mlrp}
The signal $X: \Omega \rightarrow \mathcal{P}^*$ generated by any optimal monitoring technology $\mathcal{P}^*$ satisfies strict MLRP with respect to $\overset{z}{\preceq}$, i.e., any $A, A' \in \mathcal{P}^*$ satisfy $A \overset{z}{\preceq} A'$ if and only if $z\left(A\right) < z\left(A'\right)$.
\end{cor}

While the signal generated by any monitoring technology trivially satisfies the weak MLRP with respect to $\overset{z}{\preceq}$ (i.e., replace ``$<$'' with ``$\leq$'' in Corollary \ref{cor_mlrp}), it violates the strict MLRP if there are multiple performance categories that attain the same $z$-value. By contrast, the signal generated by any optimal monitoring technology must satisfy the strict MLRP with respect to $\overset{z}{\preceq}$, because merging performance categories of the same $z$-value saves the monitoring cost while leaving the incentive cost unaffected. 

\paragraph{Comparative statics} The parameter $\mu$ captures factors that affect the (opportunity) cost of data processing and analysis. Factors that reduce $\mu$ include, but are not limited to: the advent of IT-based HR management systems in the 90's, advancements in speech analytics, increases in computing power, etc..

To facilitate comparative statics analysis, we write any choice of optimal incentive contract as $\langle \mathcal{P}^*(\mu), w^*(\cdot; \mu) \rangle$ to make its dependence on $\mu$ explicit: 

\begin{prop}\label{prop_mu}
Fix any $0<\mu<\mu'$. For any choices of $\langle \mathcal{P}^*(\mu), w^*(\cdot; \mu) \rangle$ and $\langle \mathcal{P}^*\left(\mu'\right), w^*\left(\cdot; \mu'\right) \rangle$:
\begin{enumerate}[(i)]
\item $\displaystyle \sum_{A \in \mathcal{P}\left(\mu\right)} P_1\left(A\right)w^*\left(A; \mu\right) \leq \displaystyle \sum_{A \in \mathcal{P}\left(\mu'\right)} P_1\left(A\right)w^*\left(A; \mu'\right)$;
\item $H\left(\mathcal{P}^*\left(\mu\right), 1\right) \geq H\left(\mathcal{P}^*\left(\mu'\right), 1\right)$;
\item $|\mathcal{P}^*(\mu)|\geq |\mathcal{P}^*(\mu')|$ under Assumption \ref{assm_cf}(a). 
\end{enumerate}
\end{prop}

\begin{proof}
Part (i) follows from the optimalities of $\mathcal{P}^*(\mu)$ and $\mathcal{P}^*(\mu')$. Parts (ii) and (ii) are immediate.
\end{proof}

Proposition \ref{prop_mu} shows that as data processing and analysis become cheaper, the principal pays less wage on average and the information carried by the output signal becomes finer. In the case where the monitoring cost is an increasing function of the rating scale (see, e.g., \cite{introhr}), the optimal rating scale is nonincreasing in $\mu$. For other monitoring cost functions such as entropy, we can first compute the cutoff $z$-values and then the optimal rating scale as in Example \ref{exm_main}.\footnote{In general, this is not an easy task because perturbations of cutoff $z$-values (which differ from the perturbation considered in Section \ref{sec_proofsketch}) affect wages endogenously through the Lagrange multipliers of the incentive constraints. } Figure \ref{figure_ratingscale} plots the numerical solutions obtained in a special case. 

The above findings are consistent with several strands of empirical facts. Among others, access to IT has proven to increase the fineness of the performance grids among manufacturing companies, holding other things constant (\cite{bloomsurvey}; \cite{it}).\footnote{See the appendices of \cite{bloomvanreenenwp} for survey questions regarding the fineness of the performance grids, e.g., ``Each employee is given a red light (not performing), an amber light (doing well and meeting targets), a green light (consistently meeting targets, very high performer) and a blue light (high performer capable of promotion of up to two levels),''  versus ``rewards is based on an individual's commitment to the company measured by seniority.'' } Crowdsourcing the processing and analysis of real-time data has enabled the ``exact individual diagnosis'' that separates distinctive and mediocre performers in companies like GE and Zalando (\cite{mckinsey}). 

\begin{figure}[!h]
    \centering
     \includegraphics[scale=.6]{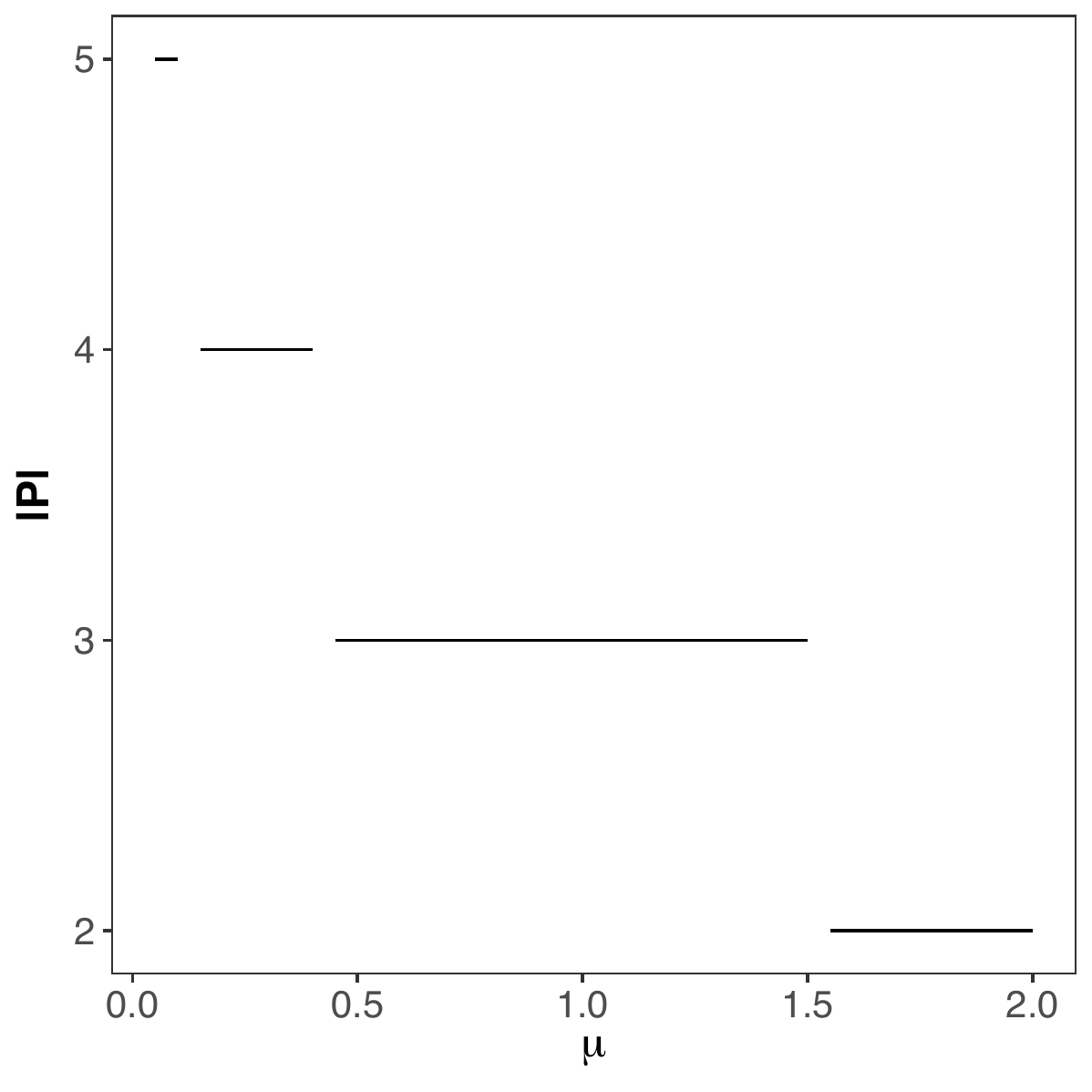}
    \caption{Plot the optimal rating scale against $\mu$: entropy cost, $u\left(w\right) =\sqrt{w}$, $Z \sim U\left[-1/2, 1/2\right]$, $c=1$, $K=100$.}\label{figure_ratingscale}
\end{figure}

\section{Extension: Multiple Agents}\label{sec_multiagent}
\subsection{Setup}\label{sec_multiagent_setup}
Each of the two agents $i=1,2$ earns a payoff $u_i\left(w_i\right)-c_i\left(a_i\right)$ from spending a nonnegative wage $w_i \geq 0$ and privately exerting high or low effort $a_i \in \{0,1\}$. The function $u_i:\mathbb{R}_+\rightarrow \mathbb{R}$ satisfies $u_i\left(0\right)=0$, $u_i'>0$ and $u''_i<0$, and $c_i\left(1\right)=c_i>c_i\left(0\right)=0$. 

Each effort profile ${\bf{a}}=a_1 a_2$ generates a probability space $\left(\Omega, \Sigma, P_{{\bf{a}}}\right)$, where $\Omega$ is a finite-dimensional Euclidean space that comprises agents' performance data, $\Sigma$ is the Borel sigma-algebra on $\Omega$, and $P_{{\bf{a}}}$ is the probability measure on $\left(\Omega, \Sigma\right)$. $P_{{\bf{a}}}$'s are assumed to be mutually absolutely continuous, and the probability density function $p_{{\bf{a}}}$'s they induce are well-defined and everywhere positive. 

In this new setting, a monitoring technology $\mathcal{P}$ can be any partition of $\Omega$ with at most $K$ cells that are all of positive measures, and a wage scheme ${\bf{w}}: \mathcal{P} \rightarrow \mathbb{R}_{+}^2$ maps each cell $A$ of $\mathcal{P}$ to a vector ${\bf{w}}\left(A\right)=\left(w_1\left(A\right), w_2\left(A\right)\right)^{\top}$ of nonnegative wages. For any data point $\omega$, let $A\left(\omega\right)$ be the unique performance category that contains $\omega$ and let ${\bf{w}}\left(A\left(\omega\right)\right)$ be the wage vector associated with $A\left(\omega\right)$. Time evolves as follows: 
\begin{enumerate}
\item the principal commits to $\langle \mathcal{P}, {\bf{w}}\left(\cdot\right)\rangle$; 
\item agent $i$ privately chooses $a_i \in \left\{0,1\right\}$, $i=1,2$;
\item Nature draws $\omega$ from $\Omega$ according to $P_{{\bf{a}}}$;
\item the monitoring technology outputs $A\left(\omega\right)$;
\item the principal pays $w_i\left(A\left(\omega\right)\right)$ to agent $i=1, 2$.
\end{enumerate}

Consider the problem of inducing both agents to exert high effort. Write $\bm{1}$ for $\left(1,1\right)^{\top}$ and define a vector-valued random variable ${\bf{Z}}=\left(Z_1, Z_2\right)^{\top}$ by 
\[
Z_i\left(\omega\right)=1-\frac{p_{a_i=0, a_{-i}=1}\left(\omega\right)}{p_{\bm{1}}\left(\omega\right)} \text{ }  \forall \omega \in \Omega, i=1,2. 
\]
Define the $\bf{z}$-value of any set $A \in \Sigma$ of positive measure by $\left(z_1\left(A\right), z_2\left(A\right)\right)^{\top}$, where 
\[
z_i\left(A\right)=\mathbf{E}\left[Z_i \mid A; {\bf{a}}={\bf{1}}\right] \text{ } \forall  i=1,2.
\]
A contract is incentive compatible for agent $i$ if
\begin{equation}
\tag{IC$_i$} \sum_{A \in \mathcal{P}} P_{\bm{1}}(A)u_i\left(w_i\left(A\right)\right)z_i\left(A\right) \geq c_i, \label{eqn_ici}
\end{equation}
and it satisfies agent $i$'s limited liability constraint if 
\begin{equation}
\tag{LL$_i$} w_i\left(A\right) \geq 0\text{ } \forall A \in \mathcal{P}. \label{eqn_lli}
\end{equation}
An optimal contract minimizes the total implementation cost under the high effort profile, subject to agents' incentive compatibility constraints and limited liability constraints: 
\[
\min_{\langle \mathcal{P}, {\bf{w}}(\cdot) \rangle} \sum_{A \in \mathcal{P}} P_{{\bf{a}}}(A) \sum_{i=1}^2 w_i(A) + \mu \cdot H(\mathcal{P}, {\bf{1}}) \text{ s.t. (\ref{eqn_ici}) and (\ref{eqn_lli}), } i=1, 2.
\]

\subsection{Analysis}\label{sec_multiagent_analysis}
The next definition generalizes $Z$-convexity:
\begin{defn}
A set $A \in \Sigma$ is ${\bf{Z}}$-convex if the following holds for all $\omega', \omega'' \in A$ such that ${\bf{Z}}\left(\omega'\right) \neq {\bf{Z}}\left(\omega''\right)$:
\[
\left\{\omega \in \Omega: {\bf{Z}}\left(\omega\right)=(1-s) \cdot {\bf{Z}}\left(\omega'\right)+s \cdot {\bf{Z}}\left(\omega''\right) \text{ for some } s \in (0,1) \right\}\subset A.
\]
\end{defn}

The next two assumptions impose regularities on the principal's problem analogously to Assumptions \ref{assm_regular} and \ref{assm_compact}: 

\begin{assm}\label{assm_regular_multiagent}
$\bf{Z}$ is distributed atomelessly on a connect set ${\bf{Z}}\left(\Omega\right)$ in $\mathbb{R}^2$ under ${\bf{a}}={\bf{1}}$.   
\end{assm}

\begin{assm}\label{assm_compact_multiagent}
${\bf{Z}}\left(\Omega\right)$ is compact set in $\mathbb{R}^2$ with $\dim {\bf{Z}}\left(\Omega\right)=2$. 
\end{assm}

The next theorems extend Theorems \ref{thm_main}  and \ref{thm_exist} to encompass multiple agents: 

\begin{thm}\label{thm_multiagent}
Assume Assumptions \ref{assm_mc}, \ref{assm_regular_multiagent} and \ref{assm_compact_multiagent}. Then any optimal monitoring technology  comprises $\bf{Z}$-convex cells that constitute convex polygons in $\mathbb{R}^2$. 
\end{thm}

\begin{thm}\label{thm_exist_multiagent}
An optimal incentive contract that induces high effort from both agents exists under Assumptions \ref{assm_mc}, \ref{assm_cf}, \ref{assm_regular_multiagent}  and  \ref{assm_compact_multiagent}. 
\end{thm}

\begin{proof}
See Appendix \ref{sec_proof_multiagent}. 
\end{proof}

\paragraph{Proof sketch} The proof strategy developed in Section \ref{sec_proofsketch} is useful for handling vector-valued $z$-values and wages. As before, fix any $\epsilon>0$, and take any subsets $A'_\epsilon$ and $A''_\epsilon$ of two distinct performance categories $A_j$ and $A_k$, respectively, such that $P_{\bf{1}}\left(A'_\epsilon\right)=P_{\bf{1}}\left(A''_\epsilon\right)=\epsilon$ and ${\bf{z}}\left(A'_\epsilon\right) = {\bf{z}}' \neq {\bf{z}}\left(A''_\epsilon\right) = {\bf{z}}''$ (Lemma \ref{lem_smallSet_multiagent} of Appendix \ref{sec_proof_multiagent_lemma} proves existence of sets that satisfy weaker properties). Post the perturbation as in Section \ref{sec_proofsketch}, the principal's Lagrangian becomes (ignore  (\ref{eqn_lli}) constraints):
\[\mathcal{L}\left(\epsilon\right)=\sum_{n}\pi_n\left(\sum_i w_{i,n}\left(\epsilon\right)-\lambda_i\left(\epsilon\right)u_i\left(w_{i,n}\left(\epsilon\right)\right) z_{i,n}\left(\epsilon\right)-c_i\right),\]
where $\pi_n$ denotes the probability of $A_n$ (equivalently, $A_n\left(\epsilon\right)$) under ${\bf{a}}={\bf{1}}$, $w_{i,n}\left(\epsilon\right)$ agent $i$'s optimal wage at $A_n\left(\epsilon\right)$ and $\lambda_i\left(\epsilon\right)$ the Lagrange multiplier associated with the (\ref{eqn_ici}) constraint. Assuming differentiability, we obtain 
\[\mathcal{L}'\left(0\right)=-\sum_{n} \pi_n \cdot  {\bf{u}}_n^{\top}  \begin{pmatrix}
    \lambda_1\left(0\right)  & 0\\
    0      &  \lambda_2\left(0\right) 
\end{pmatrix} \frac{d}{d\epsilon} {\bf{z}}_n\left(\epsilon\right)\bigg\vert_{\epsilon=0}=\left({\bf{u}}_k-{\bf{u}}_j\right)^{\top} \left(\widehat{\bf{z}}''-\widehat{\bf{z}}'\right), \]
where \[{\bf{u}}_n=\left(u_1\left(w_{i,n}\left(0\right)\right), u_2\left(w_{i,n}\left(0\right)\right)\right)^{\top}  \text{ } \forall n\]
and 
\[\widehat{\bf{z}}= \begin{pmatrix}
    \lambda_1\left(0\right)  & 0\\
    0      &  \lambda_2\left(0\right) 
\end{pmatrix} {\bf{z}} \text{ for } {\bf{z}}={\bf{z}}', {\bf{z}}''.\]
Since $\mathcal{L}'(0) \geq 0$ by optimality, the assignment of the Lagrange multiplier-weighted ${\bf{z}}$-values into performance categories must be  ``positive assortative,'' where the direction of sorting is given by the vector of agents' utilities. This implies $\bf{Z}$-convexity for the same reason as in  Section \ref{sec_proofsketch}.

\paragraph{Implications} Solving the optimal convex polygons is computationally hard. That said, note that the boundaries of convex polygons consist of straight line segments in ${\bf{Z}}\left(\Omega\right)$, which combined with Assumption \ref{assm_regular_multiagent} yields the following observations: 
\begin{itemize}
\item any bi-partitional contract takes the form of either a team or a tournament and is fully captured by the intercept and slope of the straight line as depicted in Figure \ref{figure_group};

\bigskip

\begin{figure}[htp!]
\centering
\begin{minipage}[b]{0.48\linewidth}
  \includegraphics[scale=.43]{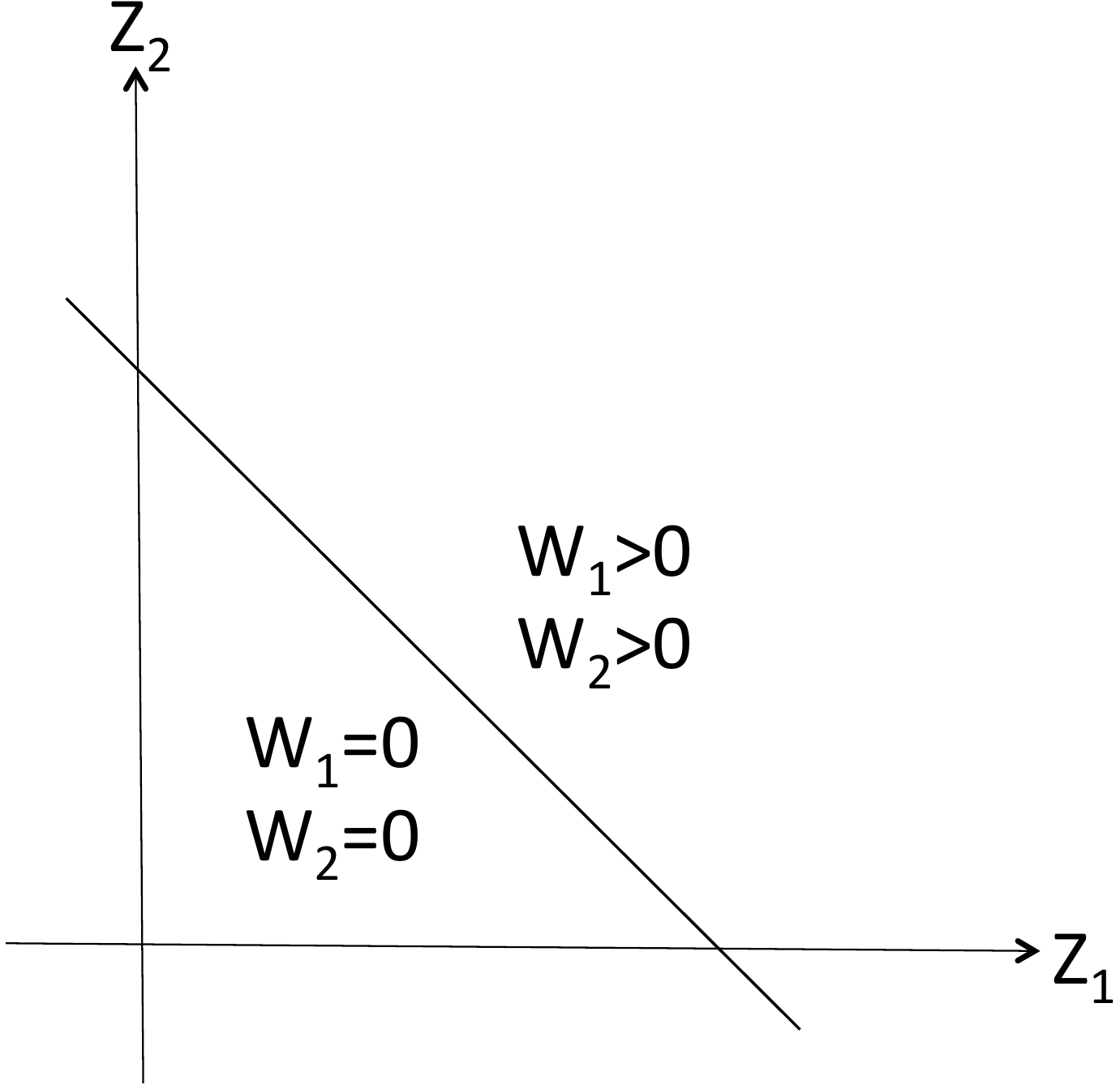}
\end{minipage}
\begin{minipage}[b]{0.48\linewidth}
  \includegraphics[scale=.43]{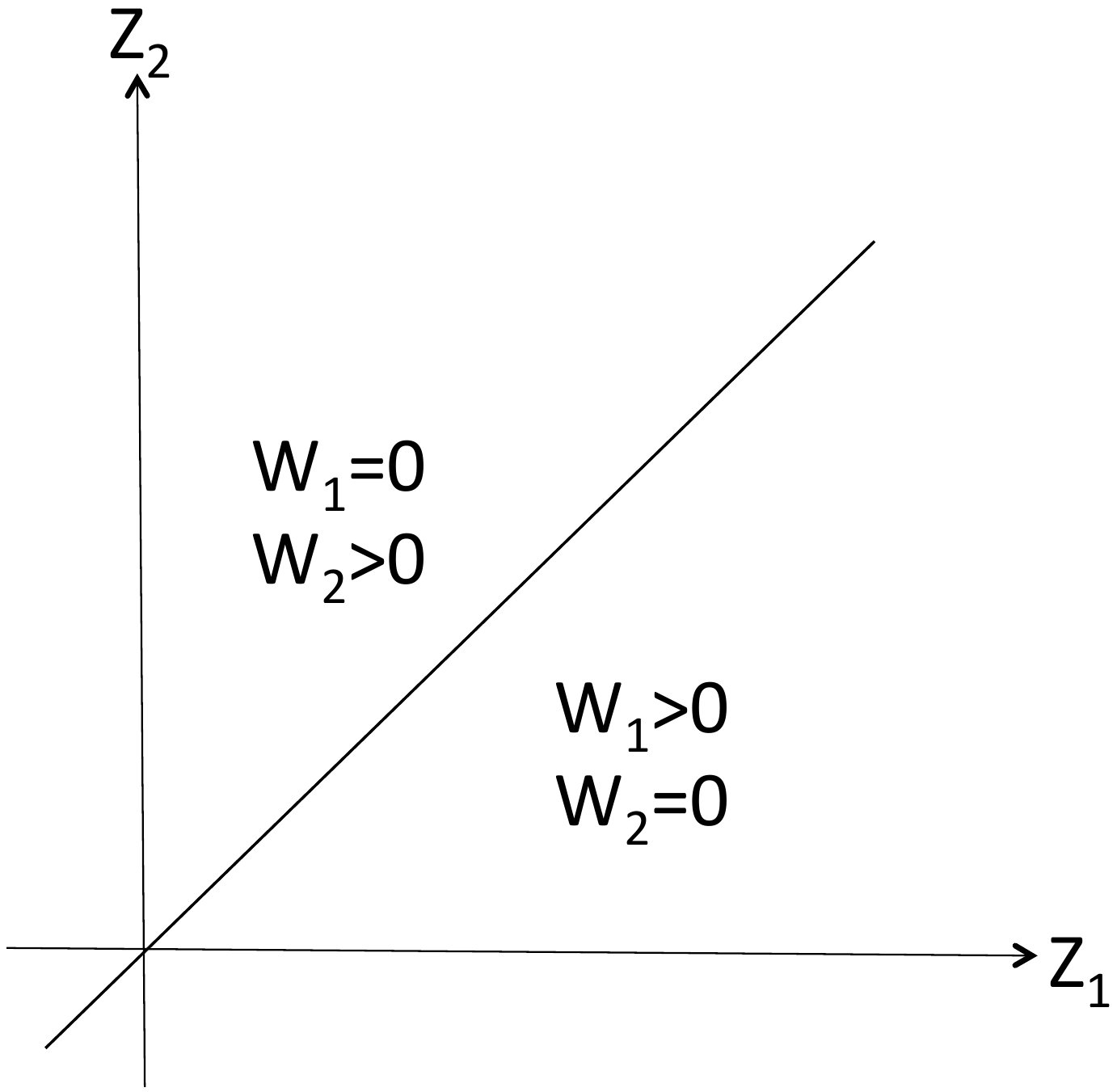}
\end{minipage}
\caption{Bi-partitional contracts: team and tournament. }
\label{figure_group}
\end{figure}

\item contracts that evaluate and reward agents on an individual basis are fully determined by the individual performance cutoffs as depicted in Figure \ref{figure_individual}.

\bigskip

\begin{figure}[htp!]
\centering
\includegraphics[scale=.43]{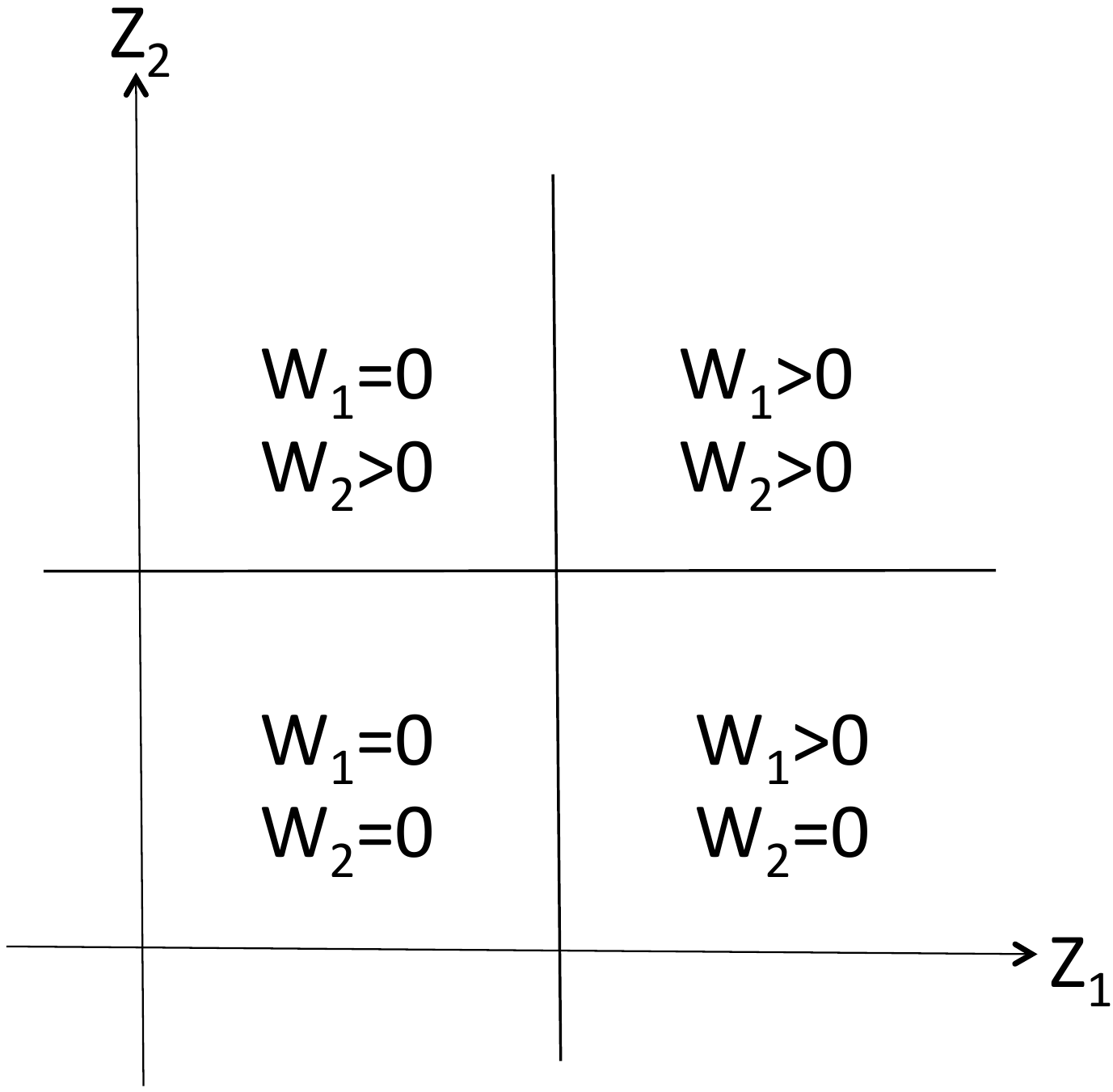}
\caption{An individual incentive contract.}
\label{figure_individual}
\end{figure}
\end{itemize}  

\subsection{Application: Individual vs. Group Evaluation}\label{sec_multiagent_application}
This section compares individual and group performance evaluations from the angle of monitoring cost. To obtain the sharpest insights, suppose that agents are \emph{technologically independent}: 

\begin{assm}\label{assm_product}
There exist probability spaces $\left\{\left(\Omega_i, \Sigma_i, P_{i, a_i}\right)\right\}_{i, a_i}$ as in Section \ref{sec_baseline} such that $\left(\Omega, \Sigma, P_{{\bf{a}}}\right)=\left(\Omega_1 \times \Omega_2, \Sigma_1 \otimes \Sigma_2, P_{1, a_1} \times P_{2, a_2}\right)$ for all ${\bf{a}} \in \left\{0,1\right\}^2$.
\end{assm}

In the language of contract theory, Assumption \ref{assm_product} rules out any  \emph{technological link} (i.e., $\omega_i$ depends on $a_{-i}$) or \emph{common productivity shock} (i.e., $\omega_1, \omega_2$ are correlated given $\bf{a}$) between agents. 

The next definition is standard: 

\begin{defn}\label{defn_individualgroup}
\begin{enumerate}[(i)]
\item $\mathcal{P}$ is an \emph{individual monitoring technology} if for all $A \in \mathcal{P}$, there exist $A_1 \in \Sigma_1$ and $A_2 \in \Sigma_2$ such that $A=A_1 \times A_2$; otherwise $\mathcal{P}$ is a \emph{group monitoring technology}; 

\item Let $\mathcal{P}$ be any individual monitoring technology. Then ${\bf{w}}: \mathcal{P} \rightarrow \mathbb{R}_{+}^2$ is an \emph{individual wage scheme} if $w_i\left(A_i \times A_{-i}'; \mathcal{P}\right)=w_i\left(A_i \times A_{-i}''; \mathcal{P}\right)$ for all $i=1, 2$ and $A_i \times A_{-i}', A_i \times A_{-i}'' \in \mathcal{P}$; otherwise ${\bf{w}}: \mathcal{P} \rightarrow \mathbb{R}_{+}^2$ is a \emph{group wage scheme};

\item $\langle \mathcal{P}, {\bf{w}}: \mathcal{P} \rightarrow \mathbb{R}_{+}^2\rangle$ is an \emph{individual incentive contract} if $\mathcal{P}$ is an individual monitoring technology and ${\bf{w}}: \mathcal{P} \rightarrow \mathbb{R}_{+}^2$ is an individual wage scheme; otherwise it is a \emph{group incentive contract}. 
\end{enumerate}
\end{defn}

By definition, a group incentive contract either conducts group performance evaluations or pairs individual performance evaluations with group incentive pays. Under Assumption \ref{assm_product}, the second option is sub-optimal by the sufficient statistics principle or \cite{holmstromteam}, thus reducing the comparison between individual and group incentive contracts to that of  individual and group performance evaluations.  

Let $I$ be the ratio between the minimal cost of implementing bi-partitional incentive contracts and that of implementing individual incentive contracts (the latter, by definition, have at least four performance categories). $I<1$ is a definitive indicator that group evaluation is optimal whereas individual evaluation is not. The next result is immediate: 

\begin{cor}\label{cor_multiagent}
Under Assumptions \ref{assm_mc}, \ref{assm_cf}(a), \ref{assm_regular_multiagent},  \ref{assm_compact_multiagent} and \ref{assm_product}, $I<1$ when $\mu$ is large.
\end{cor}

Beyond the case considered in Corollary \ref{cor_multiagent}, we can compute $I$ numerically based on the prior discussion about how to parameterize bi-partitional and individual incentive contracts. Figure \ref{figure_indicator} plots the solutions obtained in a special case.

\bigskip

\begin{figure}[h!]
\centering
\includegraphics[scale=.6]{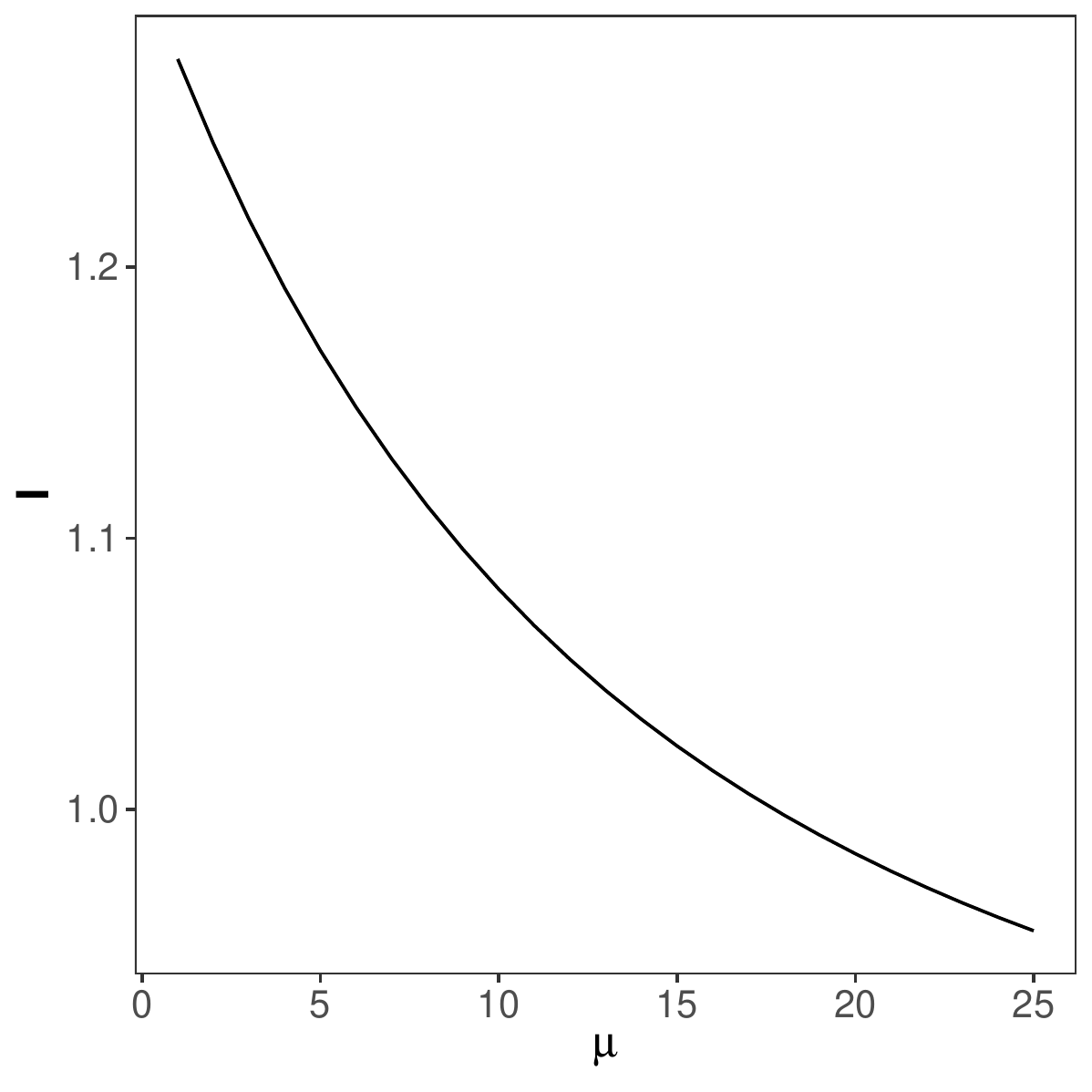}
\caption{Plot $I$ against $\mu$: entropy cost, $u_i\left(w\right)=\sqrt{w}$, $Z_i \sim U\left[-1/2, 1/2\right]$ and $c_i=1$ for $i=1,2$.}
\label{figure_indicator}
\end{figure}

Our result formalizes the theses of \cite{ad} and \cite{lazearrosen} that either team or tournament should be the dominant incentive system when individual performance evaluation is too costly to conduct. It enriches the analyses of \cite{holmstromteam}, \cite{greenstokey} and \cite{mookherjee}, which attribute the use of group incentive contracts to the technological dependence between agents while abstracting away from the issue of data processing and analysis. Recently, these views are reconciled by \cite{bloomvanreenenwp}, which find--just as our theory predicts--that companies make different choices between individual and group evaluations despite being technologically similar, and that  group evaluation is most prevalent when the capacity to sift out individual-level information is limited by, e.g., the lack of IT access.\footnote{See the survey questions of \cite{bloomvanreenenwp} regarding the choices between individual and group evaluations, e.g., ``employees are rewarded based on their individual contributions to the company,'' and ``compensation is based on shift/plant-level outcomes.'' The former is regarded as an advanced but expensive managerial practice and is more prevalent among companies with better IT access, other things being equal.} In the future, it will be interesting to nail down the role of IT in \cite{bloomvanreenenwp}, and to replicate these studies for recent advancements in data technologies. 

\section{Extension: Multiple Actions}\label{sec_multidev} 
In this section, suppose that the agent's action space $\mathcal{A}$ is a finite set, and that taking an action $a$ in $\mathcal{A}$ incurs a cost $c\left(a\right)$ to the agent and generates a probability space $\left(\Omega, \Sigma, P_a\right)$ as in Section \ref{sec_baseline}. The principal wishes to induce the most costly action $a^*$, i.e., $c\left(a^*\right)>c\left(a\right)$ for all $a \in \mathcal{D}=\mathcal{A}-\left\{a^*\right\}$. For any deviation from $a^*$ to $a \in \mathcal{D}$, define a random variable $Z_a: \Omega \rightarrow \mathbb{R}$ by
\[
Z_a\left(\omega\right)=1-\frac{p_a\left(\omega\right)}{p_{a^*}\left(\omega\right)} \text{ } \forall \omega \in \Omega. 
\]
For any $a \in \mathcal{D}$ and set $A \in \Sigma$ of positive measure, define 
\[
z_a\left(A\right)=\mathbf{E}\left[Z_a \mid A; a^*\right].
\]
A contract is incentive compatible if for all $a \in \mathcal{D}$: 
\begin{equation}
\tag{IC$_a$} \sum_{A \in \mathcal{P}} P_{a^*}\left(A\right)u\left(w\left(A\right)\right)z_{a}\left(A\right) \geq c\left(a^*\right)-c\left(a\right).\label{eqn_ic_multidev}
\end{equation}
An optimal incentive contract $\langle \mathcal{P}^*, w^*\left(\cdot\right)\rangle$ that induces $a^*$ solves
\begin{equation*}
\min_{\langle \mathcal{P}, w(\cdot)\rangle} \sum_{A \in \mathcal{P}} P_{a^*}\left(A\right)w\left(A\right)+\mu \cdot H\left(\mathcal{P}, a^*\right) \text{ s.t. (\ref{eqn_ic_multidev}) $\forall a \in \mathcal{D}$ and (\ref{eqn_ll})}. \label{eqn_problem_multidev}
\end{equation*}

Write ${\bf{Z}}$ for $\left(Z_a\right)^{\top}_{a \in \mathcal{D}}$. For any $|\mathcal{D}|$-vector $\bm{\lambda}=\left(\lambda_a\right)^{\top}_{a \in \mathcal{D}}$ in $\mathbb{R}_{+}^{|\mathcal{D}|}$, define a random variable $Z_{{\bm{\lambda}}}: \Omega \rightarrow \mathbb{R}$ by 
\[
Z_{\bm{\lambda}}\left(\omega\right)=\bm{\lambda}^{\top} {\bf{Z}}\left(\omega\right) \text{ } \forall \omega \in \Omega.
\]
The next definition generalizes $Z$-convexity:
\begin{defn}
A set $A \in \Sigma$ is $Z_{\bm{\lambda}}$-convex if the following holds for all  $\omega', \omega'' \in A$ such that $Z_{\bm{\lambda}}\left(\omega'\right) \neq Z_{\bm{\lambda}}\left(\omega''\right)$: 
\[
\left\{\omega: Z_{\bm{\lambda}}\left(\omega\right)=\left(1-s\right)\cdot  Z_{\bm{\lambda}}\left(\omega'\right)+s \cdot Z_{\bm{\lambda}}\left(\omega''\right)  \text{ for some } s \in \left(0,1\right)\right\} \subset A.
\]
\end{defn}
The next theorems extend Theorems \ref{thm_main} and \ref{thm_exist} to encompass multiple actions: 

\begin{thm}\label{thm_multidev}
Assume Assumption \ref{assm_mc} and Assumption \ref{assm_compact} for all $a \in \mathcal{D}$. Then for any optimal incentive contract $\langle \mathcal{P}^*, w^*\left(\cdot\right)\rangle$ that induces $a^{\ast}$, there exists $\bm{\lambda^*} \in \mathbb{R}_{+}^{|\mathcal{D}|}$ with $\|\bm{\lambda^*}\|>0$\footnote{$\|\cdot\| $ denotes the sup norm in the remainder of this paper.} such that all cells of $\mathcal{P}^*$ are $Z_{\bm{\lambda}^*}$-convex and can be labeled as $A_1,\cdots, A_N$ such that $0=w^*\left(A_1\right)<\cdots <w^*\left(A_N\right)$. Assume, in addition, Assumption  \ref{assm_regular} for all $a \in \mathcal{D}$. Then there exist $-\infty\leq \widehat{z}_0<\cdots< \widehat{z}_N<+\infty$ such that $A_n=\left\{\omega: Z_{\bm{\lambda}^*}\left(\omega\right) \in [\widehat{z}_{n-1}, \widehat{z}_n)\right\}$ for $n=1,\cdots, N$.
\end{thm}

\begin{thm}\label{thm_exist_multidev}
Assume Assumptions \ref{assm_mc} and \ref{assm_cf}, as well as  Assumptions \ref{assm_regular} and \ref{assm_compact} for all $a \in \mathcal{D}$. Then an optimal incentive contract that induces $a^*$ exists.
\end{thm}

\begin{proof}
See Appendix \ref{sec_proof_multidev}. 
\end{proof}

In the presence of multiple actions, each data point is associated with finitely many $z$-values, each corresponding to a deviation from $a^*$ that the agent can potentially commit.  By establishing that the assignment of Lagrange multiplier-weighted $z$-values into wage categories is positive assortative, Theorem \ref{thm_multidev} relates the focus of data processing and analysis to the agent's endogenous tendencies to commit deviations. Intuitively, when $\lambda_a^*$ is large and hence the agent is tempted to commit deviation $a$, focus should be given to the information $Z_a$ that helps detect deviation $a$, and the final performance rating should vary significantly with the assessment of $Z_a$. The next section gives an application of this result.

\subsection{Application: Multiple Tasks}\label{sec_multitask_application}
A single agent can exert either high or low effort $a_i\in \{0,1\}$ in each of the two tasks $i=1,2$, and each $a_i$ independently generates a probability space $\left(\Omega_i, \Sigma_i, P_{i, a_i}\right)$ as in Section \ref{sec_baseline}. The goal of a risk-neutral principal is to induce high effort in both tasks. 

Write ${\bf{a}}=a_1a_2$, $\bm{\omega}=\omega_1\omega_2$, $\mathcal{A}=\left\{11,01,10,00\right\}$, ${\bf{a}}^*=11$ and $\mathcal{D}=\left\{01, 10, 00\right\}$. For any $i=1,2$ and $\omega_i \in \Omega_i$, define 
\[Z_i\left(\omega_i\right)=1-\frac{p_{i, a_i=0}\left(\omega_i\right)}{p_{i, a_i=1}\left(\omega_i\right)},\] where $p_{i,a_i}$ is the probability density function induced by $P_{i,a_i}$. For any $\bm{\omega} \in 
\Omega_1 \times \Omega_2$ and $\bm{\lambda}=\left(\lambda_{01}, \lambda_{10}, \lambda_{00}\right)^{\top} \in \mathbb{R}_+^3$, define 
\[
Z_{\bm{\lambda}}\left(\bm{\omega}\right)=\left(\lambda_{01}+\lambda_{00}\right) \cdot Z_{1}(\omega_1)+\left(\lambda_{10}+\lambda_{00}\right) \cdot Z_{2}\left(\omega_2\right)
-\lambda_{00} \cdot Z_{1}\left(\omega_1\right) Z_{2}\left(\omega_2\right).\]
The next corollary is immediate from Theorem \ref{thm_multidev}:

\begin{cor}\label{cor_multitask}
Assume Assumption \ref{assm_mc} and Assumption \ref{assm_compact} for all $a \in \mathcal{D}$. Then for any optimal incentive contract $\langle \mathcal{P}^*, w^*\left(\cdot\right)\rangle$ that induces high effort in both tasks, there exists $\bm{\lambda^*} \in \mathbb{R}_+^{3}$ with $\lambda^*_{01}+\lambda^*_{00}$, $\lambda^*_{10}+\lambda^*_{00}>0$ such that all cells of $\mathcal{P}^*$ are $Z_{\bm{\lambda}^*}$-convex and can be labeled as $A_1,\cdots, A_N$ such that $0=w^*\left(A_1\right)<\cdots <w^*\left(A_N\right)$. Assume, in addition, Assumption  \ref{assm_regular} for all $a \in \mathcal{D}$. Then there exist $-\infty\leq \widehat{z}_0<\cdots< \widehat{z}_N<+\infty$ such that $A_n=\left\{\bm{\omega}: Z_{\bm{\lambda}^*}\left(\bm{\omega}\right) \in [\widehat{z}_{n-1}, \widehat{z}_n)\right\}$ for $n=1,\cdots, N$.
\end{cor}

In a seminal paper, \cite{holmstrommilgrom} shows that when the agent faces  multiple tasks, over-incentivizing tasks that generate precise performance data may prevent the completion of tasks that generate noisy performance data. That analysis abstracts away from monitoring costs and focuses on the power of (linear) compensation schemes. 

Corollary \ref{cor_multitask} delivers a different message: when it comes to allocating limited resources across the assessments of multiple task performances, the optimal allocation should reflect the agent's endogenous tendency to shirk each task. The usefulness of this result is illustrated by the next example: 

\begin{example}\label{exm_multitask}
A cashier faces two tasks: to scan items and to project warmth to customers. A piece of performance data consists of the scanner data recorded by the point of sale (POS) system, as well as the feedback gathered from customers. By Corollary \ref{cor_multitask}, the following ratio: 
\[
R=\frac{\lambda_{01}^*+\lambda_{00}^*}{\lambda_{10}^*+\lambda_{00}^*}
\]
captures how the principal should allocate limited resources across the assessments of skillfulness in scanning items and warmth. Intuitively, a small $R$ arises when the cashier is reluctant to project warmth to customers, in which case resources should be devoted to the assessment of warmth, and the final performance rating should depend significantly on such assessment. 

We examine how optimal resource allocation varies with the precision of raw performance data. As in \cite{holmstrommilgrom}, we assume that
\begin{itemize}
\item $\omega_i=a_i+ \xi_i$ for $i=1,2$, where $\xi_i$'s are independent normal random variables with mean zero and variance $\sigma_i^2$'s;

\item the cashier has CARA utility of consumption $u(w)=1-\exp\left(-\gamma w\right)$.
\end{itemize}
Unlike \cite{holmstrommilgrom}, we do not confine ourselves to linear wage schemes.

In the case where the monitoring cost is an increasing function of the rating scale, we compute $R$ for different values of $\sigma_1^2$, holding $\sigma_2^2=1$ and $|\mathcal{P}|=2$ fixed. Our findings are reported in Figure \ref{figure_multitask}. 
Assuming that our parameter choices are reasonable ones, we arrive at the following conclusion: as skillfulness becomes easier to measure--thanks to the advent of high quality scanner data--the cashier becomes more afraid to shirk the scanning task and less so about projecting coldness to customers; to correct the cashier's incentive, resources should be shifted towards the processing and analysis of customer feedback and away from that of scanner data. In the future, one can test this prediction by running field experiments as that of \cite{india}. For example, one can randomize the quality of scanner data among otherwise similar stores and examine the effect on resource allocation between scanner data and customer feedback. 

\bigskip

\begin{figure}[!h]
    \centering
    \includegraphics[scale=.6]{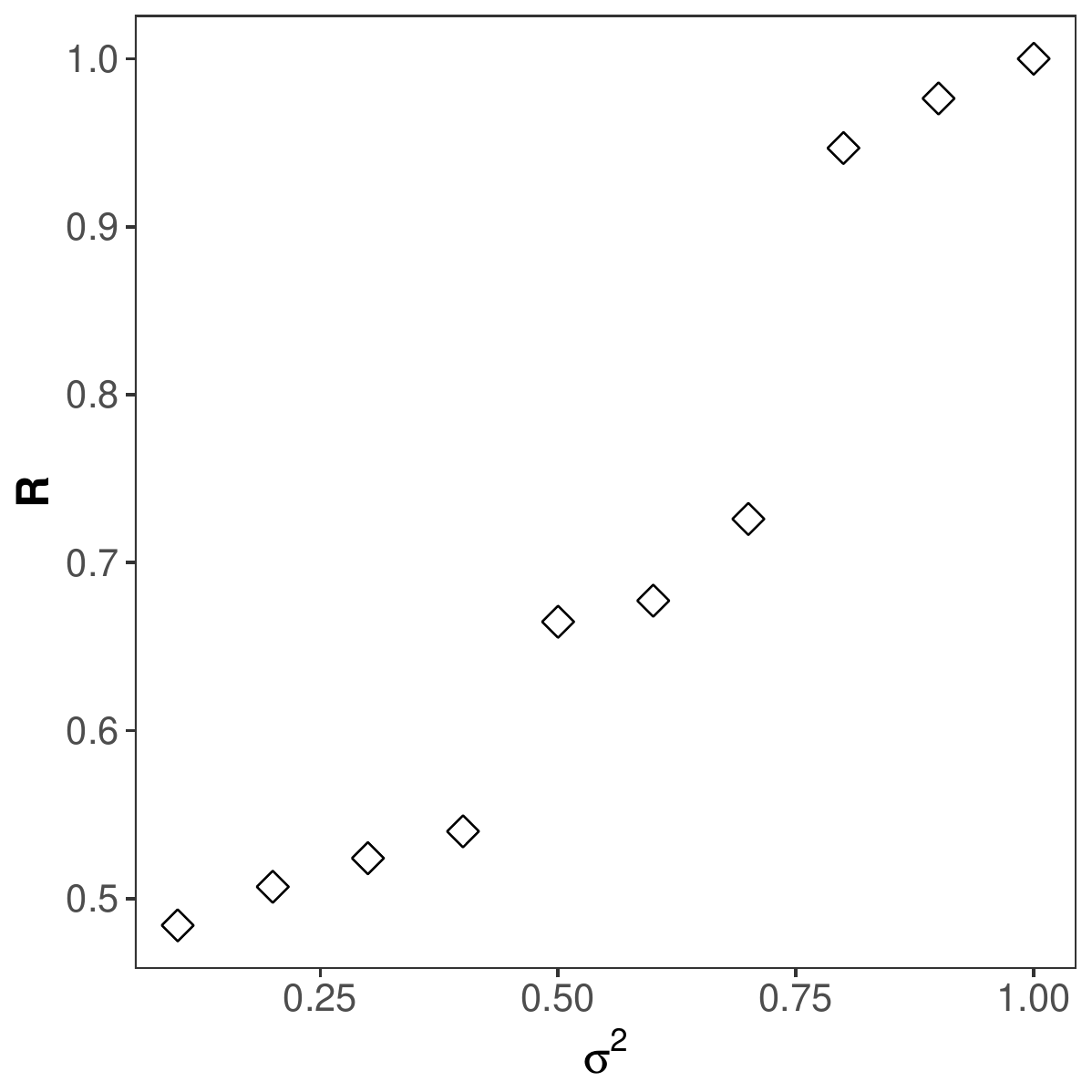}
    \caption{Plot $R$ against $\sigma_1^2$: $H\left(\mathcal{P},a\right)=f\left(|\mathcal{P}|\right)$, $|\mathcal{P}|=2$; $u\left(w\right)=1-\exp\left(-.5w\right)$; $c\left(00\right)=0$, $c\left(01\right)=0.3$, $c\left(10\right)=0.2$ and $c\left(11\right)=0.5$; $\xi_1$ and $\xi_2$ are normally distributed with mean zero and $\sigma_2^2=1$.}\label{figure_multitask}
    \end{figure}
\end{example}

\section{Conclusion}\label{sec_conclusion}
We conclude by posing open questions. First, our work is broadly related to the burgeoning literature on information design (see, e.g., \cite{bergemannmorris} for a survey), and we hope it inspires new research questions such as how to conduct costly yet flexible monitoring in long-term employment relationships. Second, our theory may guide  investigations into empirical issues such as how advancements in big data technologies have affected the design and implementation of monitoring technologies, and whether they can partially explain the heterogeneity in the internal organizations of otherwise similar firms. We hope that someone, maybe ourselves, will pursue these research agendas in the future. 

\appendix 

\section{Omitted Proofs}\label{sec_proof}

\subsection{Proofs of Section \ref{sec_analysis}}\label{sec_proof_baseline}
In this appendix, write any $N$-partitional contract $\langle \mathcal{P}, w(\cdot)\rangle$ as the corresponding tuple $\langle A_n, \pi_n, z_n, w_n\rangle_{n=1}^N$, where $A_n$ is a generic cell of $\mathcal{P}$, $\pi_n=P_1\left(A_n\right)$, $z_n=z\left(A_n\right)$ and $w_n=w\left(A_n\right)$. Assume w.l.o.g. that $z_1 \leq \cdots \leq z_N$. 

\subsubsection{Useful Lemmas}\label{sec_proof_baseline_lemma}
\noindent Proof of Lemma \ref{lem_foc}
\begin{proof}
The wage-minimization problem for given monitoring technology $\langle A_n, \pi_n, z_n\rangle_{n=1}^N$ as in Lemma \ref{lem_foc} is
\begin{align*}
\min_{\langle \tilde{w}_n\rangle}  \sum_{n=1}^N \pi_n\tilde{w}_n - \lambda \left(\sum_{n=1}^N \pi_n u\left(\tilde{w}_n\right) z_n - c\right)-\sum_{n=1}^N \eta_n \tilde{w}_n, 
\end{align*}
where $\lambda$ and $\eta_n$ denote the Lagrange multipliers associated with the (\ref{eqn_ic}) constraint and (\ref{eqn_ll}) constraint at $\tilde{w}_n$, respectively. Differentiating the objective function with respect to $\tilde{w}_n$ and setting the result equal to zero yields $\lambda z_n u'\left(w_n\right)=1-\eta_n/\pi_n,$ implying that 
$u'\left(w_n\right)=1/\left(\lambda z_n\right)$
if and only if $w_n>0$.
\end{proof}

\bigskip

\noindent Proof of Lemma \ref{lem_mlrp}
\begin{proof}
Fix any optimal incentive contract that induces high effort from the agent and let $\langle A_n, \pi_n, z_n, w_n\rangle_{n=1}^N$ be the corresponding tuple. Note  that $N \geq 2$. By Assumption \ref{assm_mc}(b), if $w_j=w_k$ for some $j \neq k$, then merging $A_j$ and $A_k$ has no effect on the incentive cost but strictly reduces the monitoring cost, which contradicts the optimality of the original contract. Then from Lemma \ref{lem_foc} and the assumption $z_1 \leq \cdots \leq z_N$, it follows that $0\leq w_1<\cdots<w_N$ and $z_1 < \cdots<z_N$. In particular, we must have $z_1<0$  because $\sum_{n=1}^N \pi_n z_n=0$. This implies $w_1=0$, because otherwise letting $w_1=0$  reduces the expected wage and relaxes the (\ref{eqn_ic}) constraint while keeping the (\ref{eqn_ll}) constraint satisfied. Finally, combining $w_n>0$ for $n \geq 2$ and Lemma \ref{lem_foc} yields $z_n>0$ for $n \geq 2$.
\end{proof}

\bigskip 

\begin{lem}\label{lem_smallSet}
For all $A\in\Sigma$ such that $P_1\left(A\right)>0$ and $\epsilon\in (0,P_1\left(A\right)]$, there exists $A_\epsilon\subset{A}$ such that  $P_1\left(A_\epsilon\right)=\epsilon$ and $z\left(A_\epsilon\right)=z\left(A\right)$. 
\end{lem}
\begin{proof}
Let $A$ be as above. Since $P_1$ admits a density, it follows that for all $t\in(0,P_1(A)]$, there exists $B_t\subset A$ such that $P_1\left(B_t\right)=t$ and $Z\left(\omega'\right)\leq Z\left(\omega\right)$ for all $\omega\in B_t$ and $\omega'\in A\setminus B_t$. Likewise, there exists $C_t\subset A$ such that  $P_1\left(C_t\right)=t$ and $Z\left(\omega'\right)\geq Z\left(\omega\right)$ for all $\omega\in C_t$ and $\omega'\in A\setminus C_t$. For $t=0$ define  $B_0=C_0=\emptyset$. 
 
Let $\epsilon$ be as above. Consider $B_t \cup C_{\epsilon-t}$, $t \in \left[0,\epsilon\right]$. Since $z\left(B_t\right)\geq z\left(A\right)$ and $z\left(C_{\epsilon-t}\right)\leq z\left(A\right)$ for all $t \in \left(0,\epsilon\right)$  and $z\left(B_t\cup C_{\epsilon-t}\right)$ is continuous in $t$ (because $P_1$ admits a density), there exists $t\in \left[0,\epsilon\right]$ such that $z\left(B_t\cup C_{\epsilon-t}\right)=z\left(A\right)$. Meanwhile $P_1\left(B_t\cup C_{\epsilon-t}\right)=\epsilon$ by construction, so let $A_{\epsilon}=B_t\cup C_{\epsilon-t}$ and we are done.
\end{proof} 

\subsubsection{Proof of Theorem \ref{thm_main}}
\begin{proof}
Take any optimal incentive contract that induces high effort from the agent and let $\langle A_n, \pi_n, z_n, w_n\rangle_{n=1}^N$ be the corresponding tuple. Suppose, to the contrary, that some $A_j$ is not $Z$-convex. By Definition \ref{defn_zconvex}, there exist $A', A'' \subset A_j$ and $ \tilde{A} \subset A_k$, $k \neq j$ such that (i) $P_1\left(A'\right)$, $P_1\left(A''\right)$,  $P_1(\tilde{A})>0$, and (ii) $\tilde{z}=\left(1-s\right)z'+sz''$, where $z'\coloneqq z\left(A'\right)\neq z'' \coloneqq z\left(A''\right)$, $\tilde{z}\coloneqq z(\tilde{A})$ and $s \in \left(0,1\right)$. By Lemma \ref{lem_smallSet}, for all $\epsilon\in (0, \min \{P_1\left(A'\right), P_1\left(A''\right), P_1(\tilde{A})\})$, there exist $A'_\epsilon \subset A'$, $A''_\epsilon \subset A''$ and $\tilde{A}_\epsilon \subset \tilde{A}$ such that (i) $P_1\left(A'_\epsilon\right)=P_1\left(A''_\epsilon\right)=P_1(\tilde{A}_\epsilon)=\epsilon$, and (ii) $z\left(A'_\epsilon\right)=z'$, $z\left(A''_\epsilon\right)=z''$ and $z(\tilde{A}_\epsilon)=\tilde{z}$. 

Consider two perturbations to the monitoring technology: (a) move $A'_\epsilon$ to $A_k$ and $\tilde{A}_\epsilon$ to $A_j$; (b) move $\tilde{A}_\epsilon$ to $A_j$ and $A_\epsilon''$ to $A_k$. By construction, neither perturbation affects the probability distribution of the output signal under high effort and hence the monitoring cost. Below we demonstrate that one of them strictly reduces the incentive cost compared to the original (optimal) contract. 

\paragraph{Perturbation (a)} Let $\langle A_n(\epsilon), \pi_n, z_n(\epsilon)\rangle_{n=1}^N$ be the tuple associated with the monitoring technology after perturbation (a). By construction, $A_j\left(\epsilon\right)=(A_j\cup \tilde{A}_{\epsilon}) \setminus A'_{\epsilon}$, so 
\[
z_j\left(\epsilon\right)=\frac{\pi_j z_j - \epsilon z' + \epsilon \tilde{z}}{\pi_j}=z_j+\displaystyle \frac{s\left(z''-z'\right)}{\pi_j} \epsilon.
\] 
Likewise, $A_k\left(\epsilon\right)=\left(A_k \cup A'_{\epsilon}\right) \setminus \tilde{A}_{\epsilon}$ and $A_n\left(\epsilon\right)=A_n$ for $n \neq j, k$, and similar algebraic manipulation as above yields
\begin{align}\label{eqn_zprime}
\begin{cases}
z_j\left(\epsilon\right)=z_j+\displaystyle \frac{s\left(z''-z'\right)}{\pi_j} \epsilon,\\
z_k\left(\epsilon\right)=z_k-\displaystyle\frac{s\left(z''-z'\right)}{\pi_k} \epsilon,\\
z_n\left(\epsilon\right)=z_n \text{ }\forall n \neq j, k. 
\end{cases}
\end{align}
Consider wage profile $\langle w_n\left(\epsilon\right)\rangle_{n=1}^N$ such that $w_1\left(\epsilon\right)=0$ and the (\ref{eqn_ic}) constraint remains binding after the perturbation, i.e., 
\begin{equation}\label{eqn_perturbic}
\sum_{n=1}^N \pi_n u\left(w_n\left(\epsilon\right)\right)z_n\left(\epsilon\right) =\sum_{n=1}^N \pi_n u\left(w_n\right)z_n=c. 
\end{equation} 
A close inspection of Equations (\ref{eqn_zprime}) and (\ref{eqn_perturbic}) reveals the existence of $M>0$ independent of $\epsilon$ such that when $\epsilon$ is small, there exist wage profiles as above that satisfy $|w_n\left(\epsilon\right)-w_n|<M\epsilon$ for all $n$ and hence the (\ref{eqn_ll}) constraint by Lemma \ref{lem_mlrp}.\footnote{To be precise, recall that $u\left(w_n\right)$, $z_n > 0$ for $n \geq 2$ by Lemma \ref{lem_mlrp}. Thus when $\epsilon$ is small, $z_n\left(\epsilon\right)> 0$ for $n \geq 2$ and solving $\sum_{n=2}^N \pi_n u\left(x_n\right) z_n\left(\epsilon\right) =\sum_{n=2}^N \pi_n u\left(w_n\right)z_n$ yields wage profiles as above. } 

With a slight abuse of notation, write $\dot{w}_{n}\left(\epsilon\right)= \left(w_{n}\left(\epsilon\right)-w_{n}\right)/\epsilon$ and $\dot{z}_{n}\left(\epsilon\right)= \left(z_{n}\left(\epsilon\right)-z_{n}\right)/\epsilon$,\footnote{Note that we do not assume differentiability of $w_n\left(\epsilon\right)$ or $z_n\left(\epsilon\right)$ with respect to $\epsilon$. The same disclaimer applies to the remainder of this paper.} and note that $\dot{w}_1\left(\epsilon\right)=0$. When $\epsilon$ is small, expanding Equation  (\ref{eqn_perturbic}) using the twice-differentiability of $u\left(\cdot\right)$ and $|w_n\left(\epsilon\right)-w_n| \sim \mathcal{O}\left(\epsilon\right)$ yields
\[
\sum_{n=1}^N \pi_n u\left(w_n\right)z_n=\sum_{n=1}^N \pi_n \left(u\left(w_n\right)+u'\left(w_n\right)\cdot \dot{w}_{n}\left(\epsilon\right) \cdot\epsilon + \mathcal{O}\left(\epsilon^2\right)\right)\left(z_n+\dot{z}_{n}\left(\epsilon\right) \cdot\epsilon\right). 
\]
Multiply the above equation by the Lagrange multiplier $\lambda>0$ associated with the (\ref{eqn_ic}) constraint prior to the perturbation. Rearranging yields
\[
\sum_{n=1}^N \pi_n \cdot u'\left(w_n\right) \cdot \lambda z_n \cdot \dot{w}_n\left(\epsilon\right)=-\lambda \sum_{n =1}^N u\left(w_{n}\right) \cdot \pi_n \dot{z}_{n}\left(\epsilon\right) + \mathcal{O}\left(\epsilon\right),
\]
and simplifying using $\dot{w}_1\left(\epsilon\right)=0$, $u'\left(w_n\right)=1/\left(\lambda z_n\right)$ for $n \geq 2$ (Lemmas \ref{lem_foc} and \ref{lem_mlrp}) and Equation (\ref{eqn_zprime}) yields
\begin{equation}\label{eqn_marginal impact of perturbation a}
\sum_{n=1}^N \pi_n \dot{w}_n\left(\epsilon\right)=s \left(u\left(w_k\right)-u\left(w_j\right)\right)  \left(\lambda z''-\lambda z'\right)+ \mathcal{O}\left(\epsilon\right). 
\end{equation}

\paragraph{Perturbation (b)} Repeating the above argument for perturbation (b) yields
\begin{equation}\label{eqn_marginal impact of perturbation b}
\sum_{n=1}^N \pi_n \dot{w}_n\left(\epsilon\right)=-\lambda \left(1-s\right)\left(u\left(w_k\right)-u\left(w_j\right)\right)\left(z''-z'\right)+\mathcal{O}\left(\epsilon\right).
\end{equation}
Then from $u\left(w_j\right)\neq u\left(w_k\right)$ (Lemma \ref{lem_mlrp}), $z' \neq z'' $ (by assumption) and $\lambda>0$, it follows that the right-hand side of either Equation (\ref{eqn_marginal impact of perturbation a}) or  (\ref{eqn_marginal impact of perturbation b}) is strictly negative when $\epsilon$ is small. Thus for either perturbation (a) or (b), we can construct a wage profile that incurs a lower incentive cost than the original optimal contract, and this leads to a contradiction.
\end{proof}

\subsubsection{Proof of Theorem \ref{thm_exist}}
\begin{proof}
By Theorem \ref{thm_main}, any optimal monitoring technology with at most $N \in \left\{2,\cdots, K\right\}$ cells is fully characterized by $N-1$ cutpoints $\widehat{z}_1,\cdots, \widehat{z}_{N-1}$ satisfying $\min Z\left(\Omega\right)\leq \widehat{z}_1\leq \cdots \leq \widehat{z}_{N-1} \leq \max Z\left(\Omega\right)$. Write $\widehat{\bf{z}}=\left(\widehat{z}_1, \cdots, \widehat{z}_{N-1}\right)^{\top}$. Define
\begin{equation*}
\mathcal{Z}_N=\left\{\widehat{\bf{z}}: \min Z\left(\Omega\right)\leq \widehat{z}_1 \leq \cdots \leq \widehat{z}_{N-1} \leq \max Z\left(\Omega\right)\right\}, 
\end{equation*}
equip $\mathcal{Z}_N$ with the sup norm $\| \cdot\|$, and note that $\mathcal{Z}_N$ is compact by Assumption \ref{assm_compact}. Let $W\left(\widehat{\bf{z}}\right)$ be the minimal incentive cost for inducing high effort from the agent under the monitoring technology formed by $\widehat{\bf{z}}$. Note that $W\left(\widehat{\bf{z}}\right)$ exists and is finite if and only if $\min Z\left(\Omega\right)<\widehat{z}_n<\max Z\left(\Omega\right)$ for some $n$, because then  $z\left(A\right) \not\equiv 0$ across the performance category $A$'s formed under $\widehat{\bf{z}}$, so $W\left(\widehat{\bf{z}}\right)$ can be solved by applying Lemma \ref{lem_foc}. 

We proceed in two steps.

\paragraph{Step 1} Show that $W\left(\widehat{\bf{z}}\right)$ is continuous in $\widehat{\bf{z}}$ for any given  $N \in \left\{2,\cdots, K\right\}$. 

\bigskip

Fix  any $\widehat{\bf{z}} \in \mathcal{Z}_N$ such that $W\left(\widehat{\bf{z}}\right)$ is finite. W.l.o.g. consider the case where $\widehat{z}_{n}$'s are all distinct. For sufficiently small $\delta>0$, let 
$\widehat{\bf{z}}^{\delta}$ be any element of $\mathcal{Z}_N$ such that $\|\widehat{\bf{z}}^{\delta}-\widehat{\bf{z}}\|<\delta$. Let $\pi_n$ and $z_n$ (resp. $\pi_n^{\delta}$ and $z_n^{\delta}$) denote the probability (under $a=1$) and $z$-value of $A_n=\left\{\omega: Z\left(\omega\right) \in [\widehat{z}_{n-1}, \widehat{z}_n)\right\}$ (resp. $A_n^{\delta}=\left\{\omega: Z(\omega) \in [\widehat{z}_{n-1}^{\delta}, \widehat{z}_n^{\delta})\right\}$), respectively. Let $w_n$ denote the optimal wage at $A_n$.

Fix any $\epsilon>0$, and consider the wage profile that pays $w_n+\epsilon$ at $A_n^{\delta}$ if $z_n^{\delta}>0$ and $w_n$ otherwise. By construction, this wage profile satisfies the (\ref{eqn_ll}) constraint. Under Assumptions \ref{assm_regular} and \ref{assm_compact}, it satisfies the (\ref{eqn_ic}) constraint when $\delta$ is sufficiently small:  
\[
\lim_{\delta \rightarrow 0}\sum_{n} \pi_n^{\delta} u\left(w_n+ 1_{z_n^{\delta}>0}\cdot \epsilon \right)z_n^{\delta}=\sum_{n} \pi_n u\left(w_n+1_{z_n>0} \cdot \epsilon\right)z_n>c,
\]
where the inequality holds because $\sum_{n} \pi_n z_n=0$ and $z_n \not\equiv 0$ so $z_n>0$ for some $n$. In addition, since 
\[
\lim_{\delta \rightarrow 0}\sum_{n} \pi_n^{\delta}\left(w_n+ 1_{z_n^{\delta}>0} \cdot \epsilon \right)=\sum_{n} \pi_n\left(w_n+1_{z_n>0} \cdot \epsilon\right), 
\]
it follows that when $\delta$ is sufficiently small, 
\[
W\left(\widehat{\bf{z}}^{\delta}\right)-W\left(\widehat{\bf{z}}\right)\leq \sum_{n} \pi_n^{\delta}\left(w_n+1_{z_n^\delta>0}\cdot \epsilon\right)-\sum_{n} \pi_n w_n <\epsilon, 
\]
where the first inequality holds because the constructed wage profile is not necessarily optimal under $\widehat{{\bf{z}}}^{\delta}$. Finally, interchanging the roles between $\widehat{\bf{z}}$ and $\widehat{\bf{z}}^{\delta}$ in the above derivation yields $W\left(\widehat{\bf{z}}\right)-W\left(\widehat{\bf{z}}^{\delta}\right)<\epsilon$, implying that $\left\vert W\left(\widehat{\bf{z}}^{\delta}\right)-W\left(\widehat{\bf{z}}\right)\right\vert<\epsilon$ when $\delta$ is sufficiently small.

\paragraph{Step 2} Under Assumption \ref{assm_cf}(a), 
the following quantity:
\[
W_N\coloneqq \min_{\widehat{\bf{z}} \in \mathcal{Z}_N} W\left(\widehat{\bf{z}}\right)
\]
exists and is finite for all $N \in \left\{2,\cdots, K\right\}$ by Step 1 and the compactness of $\mathcal{Z}_N$. Let $m_N$ denote the minimal rating scale attained by $W_N$. Solving \[\min_{2 \leq N\leq K}W_N+\mu \cdot f\left(m_N\right)\] yields the solution(s) to the principal's problem. 

Under Assumption \ref{assm_cf}(b), the principal's problem can be written as follows: 
\begin{equation*}
\min_{\widehat{\bf{z}} \in \mathcal{Z}_K} W\left(\widehat{\bf{z}}\right) + \mu\cdot h\left(\bm{\pi}\left(\widehat{\bf{z}}\right)\right), 
\end{equation*}
where $\bm{\pi}\left(\widehat{\bf{z}}\right)$ is the probability vector formed under $\widehat{\bf{z}}$ and is clearly continuous in $\widehat{\bf{z}}$. 
The existence of solution(s) then follows from Step 1 and the compactness of $\mathcal{Z}_{K}$. 
\end{proof}

\subsection{Proof of Section \ref{sec_multiagent}}\label{sec_proof_multiagent}
In this appendix, write any $N$-partitional contract $\langle \mathcal{P}, {\bf{w}}(\cdot) \rangle$ as the corresponding tuple $\langle A_n, \pi_n, {\bf{z}}_n, {\bf{w}}_n \rangle_{n=1}^N$, where $A_n$ is a generic cell of $\mathcal{P}$, $\pi_n=P_{{\bf{1}}}\left(A_n\right)$, ${\bf{z}}_n=\left(z_{1,n},z_{2,n}\right)^{\top}
\coloneqq \left(z_1\left(A_n\right), z_2\left(A_n\right)\right)^{\top}$ and ${\bf{w}}_n=\left(w_{1,n},w_{2,n}\right)^{\top}\coloneqq \left(w_1\left(A_n\right), w_2\left(A_n\right)\right)^{\top}$. 

\subsubsection{Useful Lemmas}\label{sec_proof_multiagent_lemma}
The next lemma generalizes Lemmas \ref{lem_foc} and \ref{lem_mlrp} to encompass multiple agents: 

\begin{lem}\label{lem_multiagent}
Assume Assumption \ref{assm_mc}. Then under any optimal incentive contract that induces high effort from both agents, (i) there exist $\lambda_1, \lambda_2>0$ such that $u_i'\left(w_{i,n}\right)=1/\left(\lambda_i z_{i,n} \right)$ if and only if $w_{i,n}>0$; (ii) ${\bf{w}}_j \neq {\bf{w}}_k$ for all $j \neq k$. 
\end{lem}

\begin{proof}
The wage-minimization problem for given monitoring technology $\langle A_n, \pi_n, {\bf{z}}_n\rangle_{n=1}^N$ is
\begin{align*}
\min_{\langle \tilde{w}_{i,n}\rangle} \sum_{i,n} \pi_n  \tilde{w}_{i,n}-\sum_{i} \lambda_i \left(\sum_{n} \pi_n u_i
\left(\tilde{w}_{i,n}\right)z_{i,n} - c_i\right) -\sum_{i,n}\eta_{i,n}\tilde{w}_{i,n},
\end{align*}
where $\lambda_i$ and $\eta_{i,n}$ denote the Lagrange multipliers associated with the (\ref{eqn_ici}) constraint and (\ref{eqn_lli}) constraint at $\tilde{w}_{i,n}$, respectively. Differentiating the objective function with respect to $\tilde{w}_{i,n}$ yields the first-order condition in Part (i). The proof of Part (ii) is the same as that of Lemma \ref{lem_mlrp} and is therefore omitted.
\end{proof}

The next lemma plays an analogous role as that of Lemma \ref{lem_smallSet}:
\begin{lem}\label{lem_smallSet_multiagent}
Assume Assumption \ref{assm_compact_multiagent}. Fix any $\delta>0$ and  any $A \in \Sigma$ such that $P_{{\bf{1}}}\left(A\right)>0$. Then for all $\epsilon \in (0, P_{{\bf{1}}}\left(A\right)]$, there exists $A_{\epsilon} \subset A$ such that $P_{{\bf{1}}}\left(A_{\epsilon}\right)=\epsilon$ and $\|{\bf{z}}\left(A_{\epsilon}\right)- {\bf{z}}\left(A\right)\|<\delta$.
\end{lem}

\begin{proof}
With a slight abuse of notation, let $\mathcal{P}$ be any finite partition of $\Omega$ such that every $B \in \mathcal{P}$ is measurable and $\|{\bf{Z}}\left(\omega\right)-{\bf{Z}}\left(\omega'\right)\|<\delta$ for all $\omega, \omega' \in B$. $\mathcal{P}$ exists because $P_{\bf{1}}$ admits a density and ${\bf{Z}}\left(\Omega\right)$ is a compact set in $\mathbb{R}^2$. Define $\mathcal{P}^+=\left\{B\in \mathcal{P}: P_{\bf{1}}\left(A\cap B\right)>0\right\}$ and $\mathcal{P}^0=\left\{B\in \mathcal{P}: P_{\bf{1}}\left(A\cap B\right)=0\right\}$, which are both finite. Note that $\sum_{B \in \mathcal{P}^0}P_{\bf{1}}\left(A \cap B\right)=0$, $\sum_{B \in \mathcal{P}^+}P_{\bf{1}}\left(A \cap B\right)=P_{\bf{1}}\left(A\right)$ and ${\bf{z}}\left(A\right)=\sum_{B \in \mathcal{P}^+} P_{\bf{1}}\left(A \cap B\right) {\bf{z}} \left(A \cap B\right)$. 

Since $P_{\bf{1}}$ admits a density, it follows that for all $B \in \mathcal{P}^+$, there exists $C_B \subset A \cap B$ such that $P_1\left(C_B\right)=P_{\bf{1}}\left(A\cap B\right)\epsilon \slash P_{\bf{1}}\left(A\right)$. Also note that $\|{\bf{z}}\left(C_B\right)-{\bf{z}}\left(A\cap B\right)\|<\delta$ by construction. Let $A_{\epsilon}=\cup_{B \in \mathcal{P}^+} C_{B}$. Then 
$P_{\bf{1}}\left(A_{\epsilon}\right)=\sum_{B \in \mathcal{P}^+}P_{\bf{1}}\left(A\cap B\right)\epsilon/P_{\bf{1}}\left(A\right)=\epsilon$ and 
\begin{multline*}
\|{\bf{z}}\left(A_{\epsilon}\right)-{\bf{z}}\left(A\right)\|
=\|\sum_{B \in \mathcal{P}^+}\frac{P_{\bf{1}}\left(A \cap B\right)}{P_{\bf{1}}\left(A\right)} \left( {\bf{z}}\left(C_{B}\right)-{\bf{z}}\left(A \cap B\right)\right)\|\\
\leq \sum_{B \in \mathcal{P}^+}\frac{P_{\bf{1}}\left(A\cap B\right)}{P_{\bf{1}}\left(A\right)} \|{\bf{z}}\left(C_{B}\right)-{\bf{z}}\left(A \cap B\right)\|
<\delta.
\end{multline*}
\end{proof}

\subsubsection{Proof of Theorem \ref{thm_multiagent}}
\begin{proof}
Take any optimal incentive contract that induces high effort from both agents  and let $\langle A_n, \pi_n, {\bf{z}}_n, {\bf{w}}_n \rangle_{n=1}^N$ be the corresponding tuple. Suppose, to the contrary, that some $A_j$ is not ${\bf{Z}}$-convex. By definition, there exist $A', A'' \subset A_j$ and $\tilde{A} \in A_k$, $k \neq j$ such that (i) $P_{{\bf{1}}}\left(A'\right)$, $P_{{\bf{1}}}\left(A''\right)$,  $P_{{\bf{1}}}(\tilde{A})>0$, and (ii) $\tilde{{\bf{z}}}=(1-s) {\bf{z}}'+s {\bf{z}}''$ where ${\bf{z}}'\coloneqq {\bf{z}}\left(A'\right) \neq {\bf{z}}''\coloneqq {\bf{z}}\left(A''\right)$, $\tilde{\bf{z}}\coloneqq {\bf{z}}(\tilde{A})$ and $s \in (0,1)$. By Lemma \ref{lem_smallSet_multiagent}, for all $\delta>0$ and $ \epsilon\in (0,\min\{ P_{{\bf{1}}}(A'), P_{{\bf{1}}}(A''), P_{{\bf{1}}}(\tilde{A})\})$, there exist $A'_\epsilon \subset A'$, $A''_\epsilon \subset A''$ and $\tilde{A}_\epsilon \subset \tilde{A}$ such that (i) $P_{{\bf{1}}}\left(A''_\epsilon\right)=P_{{\bf{1}}}\left(A''_\epsilon\right)=P_{{\bf{1}}}(\tilde{A}_\epsilon)  = \epsilon$, and (ii) $\|{\bf{z}}\left(A'_{\epsilon}\right)-{\bf{z}}'\|,$ $\|{\bf{z}}\left(A''_{\epsilon}\right)-{\bf{z}}''\| $, $\|{\bf{z}}(\tilde{A}_{\epsilon})-\tilde{\bf{z}}\|<\delta$.

Consider two perturbations to the monitoring technology: (a) move $A'_\epsilon$ to $A_k$ and $\tilde{A}_\epsilon$ to $A_j$; (b) move $\tilde{A}_\epsilon$ to $A_j$ and $A''_\epsilon$ to $A_k$. By Assumption \ref{assm_mc}, neither perturbation affects the probability distribution of the output signal under ${\bf{a}}={\bf{1}}$ and hence the monitoring cost. Below we demonstrate that one of them strictly reduces the incentive cost compared to the original optimal contract. 

\paragraph{Perturbation (a)} Let $\langle A_n\left(\epsilon\right), \pi_n, {\bf{z}}_n\left(\epsilon\right) \rangle_{n=1}^N$ denote the tuple associated with the monitoring technology after perturbation (a), where $A_j\left(\epsilon\right)=(A_j\cup \tilde{A}_{\epsilon}) \setminus A'_{\epsilon}$, $A_k\left(\epsilon\right)=\left(A_k \cup A'_{\epsilon}\right) \setminus \tilde{A}_{\epsilon}$ and $A_n\left(\epsilon\right)=A_n$ for $n \neq j, k$. Straightforward algebra shows that  
\begin{align}\label{eqn_multiagent_zprime}
\begin{cases}
 {\bf{z}}_j(\epsilon)={\bf{z}}_j+ \displaystyle \frac{{\bf{z}}(\tilde{A}_\epsilon)-{\bf{z}}\left(A'_\epsilon\right)}{\pi_j}\epsilon, \\
{\bf{z}}_k(\epsilon)={\bf{z}}_k- \displaystyle \frac{{\bf{z}}(\tilde{A}_\epsilon)-{\bf{z}}\left(A'_\epsilon\right)}{\pi_k}\epsilon, \\
{\bf{z}}_n(\epsilon)={\bf{z}}_n \text{ }\forall n \neq j, k,
\end{cases}
\end{align}
and that 
\begin{align}\label{eqn_zbound}
\|{\bf{z}}(\tilde{A}_{\epsilon})-{\bf{z}}\left(A'_\epsilon\right)-\left(\tilde{{\bf{z}}}-{\bf{z}}'\right)\| \leq & \|{\bf{z}}(\tilde{A}_{\epsilon})-\tilde{\bf{z}}\|+\|{\bf{z}}\left(A'_\epsilon \right)-{\bf{z}}'\| \nonumber \\
&<\min\left\{2\delta, 4\max_{\omega \in \Omega}\|{\bf{Z}}\left(\omega\right)\|\right\}. 
\end{align}
Define $\mathcal{B}_i=\left\{n: w_{i,n}=0\right\}$ for $i=1,2$. Consider wage profile $\langle {\bf{w}}_{n}\left(\epsilon\right)\rangle_{n=1}^N$ such that for $i=1,2$: (1) $w_{i,n}\left(\epsilon\right)=w_{i,n}=0$ for $n \in \mathcal{B}_i$; (2) agent $i$'s incentive compatibility constraint remains binding after perturbation (a), i.e., 
\begin{equation}\label{eqn_multiagent_perturbic}
 \sum_{n=1}^N \pi_n u_i\left(w_{i,n}\left(\epsilon\right)\right)z_{i,n}\left(\epsilon\right)=\sum_{n=1}^N \pi_n u_i\left(w_{i,n}\right)z_{i,n}=c_i.
\end{equation}
A close inspection of Equations (\ref{eqn_multiagent_zprime})-(\ref{eqn_multiagent_perturbic}) reveals the existence of $M>0$ independent of $\epsilon$ and $\delta$ such that when $\epsilon$ is sufficiently small, there exist wage profiles as above that satisfy $\|{\bf{w}}_n\left(\epsilon\right)-{\bf{w}}_n\|<M\epsilon$ for all $n$ and hence (\ref{eqn_lli}) constraints. 

With a slight abuse of notation, write $\dot{\bf{w}}_{n}\left(\epsilon\right)= \left({\bf{w}}_{n}\left(\epsilon\right)-{\bf{w}}_{n}\right)/\epsilon$ and $\dot{\bf{z}}_{n}\left(\epsilon\right)= \left({\bf{z}}_{n}\left(\epsilon\right)-{\bf{z}}_{n}\right)/\epsilon$, and note that $\dot{w}_{i,n}\left(\epsilon\right)=0$ for $i=1,2$ and $n \in \mathcal{B}_i$. When $\epsilon$ is small, expanding Equation (\ref{eqn_multiagent_perturbic}) using the twice-differentiability of $u_i\left(\cdot\right)$ and $|w_{i,n}\left(\epsilon\right)-w_{i,n}| \sim \mathcal{O}\left(\epsilon\right)$ and multiplying the result by the Lagrange multiplier $\lambda_i>0$ associated with the (\ref{eqn_ici}) constraint prior to the perturbation yields
\[
\sum_{n=1}^N \pi_n \cdot u_i'\left(w_{i,n}\right)  \cdot \lambda_i z_{i,n} \cdot \dot{w}_{i,n}\left(\epsilon\right) 
= -\lambda_i \sum_{n=1}^N u_i\left(w_{i,n}\right) \cdot \pi_n \dot{z}_{i,n}\left(\epsilon\right) + \mathcal{O}\left(\epsilon\right).
\]
Simplifying using $\dot{w}_{i,n}\left(\epsilon\right)=0$ if $n \in \mathcal{B}_i$, $u'\left(w_{i,n}\right)=1/\left(\lambda_i z_{i,n}\right)$ if $n \notin \mathcal{B}_i$ (Lemma \ref{lem_multiagent}) and Equation (\ref{eqn_multiagent_zprime}) yields
\[
 \sum_{i,n} \pi_n \dot{w}_{i,n}= \left({\bf{u}}_k-{\bf{u}}_j\right)^{\top} \Lambda \text{ } ({\bf{z}}(\tilde{A}_\epsilon)-{\bf{z}}\left(A'_\epsilon\right))+\mathcal{O}\left(\epsilon\right),
\]
where ${\bf{u}}_n= \left(u_1\left(w_{1,n}\right), 
 u_2\left(w_{2,n}\right)\right)^{\top}$
 for $n=k,j$ and $\Lambda = \left(\begin{smallmatrix}
    \lambda_1  & 0\\
    0      &  \lambda_2 
\end{smallmatrix}\right)$. Further simplifying using Equation (\ref{eqn_zbound}) and $\tilde{{\bf{z}}}=(1-s) {\bf{z}}'+s {\bf{z}}''$ yields the following when $\delta$ is small: 
\begin{align}\label{eqn3}
\sum_{i,n} \pi_n \dot{w}_{i,n}=&\left({\bf{u}}_k-{\bf{u}}_j\right)^{\top} \Lambda         \text{ } (\tilde{\bf{z}}-{\bf{z}}')+\mathcal{O}\left(\epsilon\right)\nonumber \\
&+\left({\bf{u}}_k-{\bf{u}}_j\right)^{\top} \Lambda \text{ } ({\bf{z}}(\tilde{A}_{\epsilon})-{\bf{z}}\left(A'_\epsilon\right)-\left(\tilde{{\bf{z}}}-{\bf{z}}'\right)) \nonumber \\
=&s \left({\bf{u}}_k-{\bf{u}}_j\right)^{\top}  \Lambda \text{ } \left({\bf{z}}''-{\bf{z}}'\right)+\mathcal{O}\left(\epsilon\right)+\mathcal{O}\left(\delta\right).
\end{align}

\paragraph{Perturbation (b)} Repeating the above argument for perturbation (b) yields
 \begin{equation}\label{eqn4}
 \sum_{i,n} \pi_n \dot{w}_{i,n}=-(1-s)\left({\bf{u}}_k-{\bf{u}}_j\right)^{\top}  \Lambda \left({\bf{z}}''-{\bf{z}}'\right)+\mathcal{O}\left(\epsilon\right)+\mathcal{O}\left(\delta\right). 
 \end{equation}
 
Consider two cases: 
\begin{enumerate}
\item[Case 1] $\left({\bf{u}}_k-{\bf{u}}_j\right)^{\top} \Lambda \left({\bf{z}}''-{\bf{z}}'\right) \neq 0$. In this case, the right-hand sides of Equations (\ref{eqn3}) and (\ref{eqn4}) have the opposite signs when $\epsilon$ and $\delta$ are sufficiently small, and the remainder of the proof is the same as that of Theorem \ref{thm_main}. 

\item[Case 2] $\left({\bf{u}}_k-{\bf{u}}_j\right)^{\top} \Lambda \left({\bf{z}}''-{\bf{z}}'\right) = 0$. In this case, note that $\left({\bf{u}}_k-{\bf{u}}_j\right)^{\top} \Lambda \neq {\bf{0}}^{\top}$ by Lemma \ref{lem_multiagent}, where $\bf{0}$ denotes the $2$-vector of zeros. Then from Assumption \ref{assm_regular_multiagent} ($\bf{Z}$ is distributed atomlessly on a connected set), there exist $B' \subset A'$, $B'' \subset A''$ and $\tilde{B} \subset \tilde{A}$ such that $P_{{\bf{1}}}\left(B'\right)$, $P_{{\bf{1}}}\left(B''\right)$, $P_{{\bf{1}}}(\tilde{B})>0$, ${\bf{z}}(\tilde{B})=\left(1-s'\right){\bf{z}}\left(B'\right)+s'  {\bf{z}}\left(B''\right)$ for some $s' \in \left(0,1\right)$, and $\left({\bf{u}}_k-{\bf{u}}_j\right)^{\top} \Lambda  \left({\bf{z}}\left(B''\right)-{\bf{z}}\left(B'\right)\right) \neq 0$. Replacing $A'$, $A''$ and $\tilde{A}$ with $B'$, $B''$ and $\tilde{B}$, respectively, in the above argument gives the desired result. 
\end{enumerate}
 \end{proof}
 
 \subsubsection{Proof of Theorem \ref{thm_exist_multiagent}}
\begin{proof}
By Theorem \ref{thm_multiagent}, any optimal monitoring technology with at most $N \in \left\{2,\cdots, K\right\}$ cells is fully characterized by (1) a finite  number $q_N$ of vertices ${\bf{z}}_1,\cdots, {\bf{z}}_{q_N}$ in ${\bf{Z}}\left(\Omega\right)$, and (2) a $q_N \times q_N$ adjacency matrix ${\bf{M}}$ whose $lm$'th entry equals $1$ if ${\bf{z}}_{l}$ and ${\bf{z}}_m$ are connected by a line segment and $0$ otherwise. By definition, ${\bf{M}}$ is symmetric and hence is determined by its upper triangle entries, which can be either 0 or 1. Thus ${\bf{M}}$ belongs to $\mathcal{M}_N
\coloneqq \left\{0,1\right\}^{q_N\times(q_N-1)/2}$, which is a finite set.

Write $\vec{\bf{z}}$ for $\left({\bf{z}}_1,\cdots, {\bf{z}}_{q_N}\right)^{\top}$. For any $N \in \left\{2,\cdots, K\right\}$ and adjacency matrix ${\bf{M}}\in \mathcal{M}_N$, define 
\begin{equation*}
\mathcal{Z}_N\left({\bf{M}}\right)=\left\{\vec{\bf{z}}: \left(\vec{\bf{z}}, {\bf{M}}\right)\text{ partitions } {\bf{Z}}\left(\Omega\right) \text{ into at most } N \text{ convex polygons} \right\}, 
\end{equation*}
equip $\mathcal{Z}_N\left({\bf{M}}\right)$ with the sup norm $\| \cdot\|$, and note that $\mathcal{Z}_N\left({\bf{M}}\right)$ is compact by Assumption \ref{assm_compact_multiagent}. Let $W\left(\vec{\bf{z}},{\bf{M}}\right)$ denote the minimal incentive cost for inducing high effort from both agents under the monitoring technology formed by $\left(\vec{\bf{z}}, {\bf{M}}\right)$. $W\left(\vec{\bf{z}},{\bf{M}}\right)$ exists and is finite if and only if for all $i=1,2$, $z_i\left(A\right) \not \equiv 0$ across the performance category $A$'s formed under $\left(\vec{\bf{z}}, {\bf{M}}\right)$.

We proceed in two steps.

\paragraph{Step 1} Show that $W\left(\vec{\bf{z}}, {\bf{M}}\right)$ is continuous in the first argument for any given $N \in \left\{2,\cdots, K\right\}$ and ${\bf{M}} \in \mathcal{M}_N$. 

\bigskip

Fix any $\vec{\bf{z}} \in \mathcal{Z}_N\left({\bf{M}}\right)$ such that $W\left(\vec{\bf{z}}, {\bf{M}} \right)$ is finite. For sufficiently small $\delta>0$, let $\vec{\bf{z}}^{\delta}$ be any element of $\mathcal{Z}_N\left({\bf{M}}\right)$ such that $\|\vec{\bf{z}}^{\delta}-\vec{\bf{z}}\|<\delta$. Label the performance categories formed under $\left(\vec{\bf{z}}, {\bf{M}}\right)$ and $\left(\vec{\bf{z}}^{\delta}, {\bf{M}}\right)$ as $A_n$'s and $A_n^{\delta}$'s, respectively, such that for $n=1,2,\cdots$, ${\bf{z}}_l$ is a vertex of $cl\left({\bf{Z}}\left(A_n\right)\right)$ if and only if ${\bf{z}}_l^{\delta}$ is a vertex of $cl\left({\bf{Z}}\left(A_n^{\delta}\right)\right)$. Let $\pi_n$ and $z_{i,n}$ (resp. $\pi_n^{\delta}$ and $z_{i,n}^{\delta}$) denote the probability (under ${\bf{a}}={\bf{1}}$) and $z_i$-value of $A_n$ (resp. $A_n^{\delta}$), respectively. Let $w_{i,n}$ denote the optimal wage of agent $i$ at $A_n$. 

Fix any $\epsilon>0$. Consider the wage profile that pays $w_{i,n}+\epsilon/2$ to agent $i$ if $z_{i,n}^{\delta}>0$ and $w_{i,n}$ otherwise and therefore  satisfies the (\ref{eqn_lli}) constraint. Under Assumptions \ref{assm_regular_multiagent} and \ref{assm_compact_multiagent}, the (\ref{eqn_ici}) constraint is satisfied when $\delta$ is sufficiently small: 
\[
\lim_{\delta \rightarrow 0}\sum_{n} \pi_{i,n}^{\delta}  u\left(w_{i,n}+ 1_{z_{i,n}^{\delta}>0} \cdot \epsilon/2 \right)z_{i,n}^{\delta}=\sum_{n} \pi_n u\left(w_{i,n}+1_{z_{i,n}>0} \cdot \epsilon/2\right) z_{i,n}>c_i,
\]
where the inequality holds because $\sum_{n} \pi_n z_{i,n}=0$ and $z_{i,n} \not \equiv 0$ so $z_{i,n}>0$ for some $n$. In addition, since 
\[
\lim_{\delta \rightarrow 0}\sum_{i, n} \pi_{n}^{\delta}\left(w_{i,n}+ 1_{z_{i,n}^{\delta}>0} \cdot \epsilon/2 \right)=\sum_{i, n} \pi_n\left(w_{i,n}+1_{z_{i,n}>0} \cdot \epsilon/2\right), 
\]
it follows that when $\delta$ is sufficiently small, 
\[
W\left(\vec{\bf{z}}^{\delta}, {\bf{M}}\right)-W\left(\vec{\bf{z}}, {\bf{M}}\right)\leq \sum_{i,n} \pi_n^{\delta}\left(w_{i,n}+1_{z_{i,n}^\delta>0}\cdot \epsilon/2\right)-\sum_{i,n} \pi_nw_{i,n} <\epsilon,
\]
where the first inequality holds because the constructed wage profile is not necessarily optimal under $\left(\vec{\bf{z}}^{\delta}, {\bf{M}}\right)$. Finally, interchanging the roles between $\vec{\bf{z}}^{\delta}$ and $\vec{\bf{z}}$ in the above derivation yields $W\left(\vec{\bf{z}}, {\bf{M}} \right)-W\left(\vec{\bf{z}}^{\delta}, {\bf{M}}\right)<\epsilon$, implying that $\left\vert W\left(\vec{\bf{z}}^{\delta}, {\bf{M}}\right)-W\left(\vec{\bf{z}}, {\bf{M}}\right)\right\vert<\epsilon$ when $\delta$ is sufficiently small.

\paragraph{Step 2} Under Assumption \ref{assm_cf}(a), the following quantity: 
\[W_N \coloneqq \min_{{\bf{M}} \in \mathcal{M}_N,  \vec{\bf{z}} \in \mathcal{Z}_N\left({\bf{M}}\right)} W\left(\vec{\bf{z}},{\bf{M}}\right)
\]
exists and is finite for all $N \in \left\{2,\cdots, K\right\}$ by Step 1, the compactness of $\mathcal{Z}_N \left({\bf{M}}\right)$ and the finiteness of $\mathcal{M}_N$. Under Assumption \ref{assm_cf}(b), the principal's problem can be written as follows: 
\[\min_{{\bf{M}} \in \mathcal{M}_K, \vec{\bf{z}} \in \mathcal{Z}_K({\bf{M}})}  W\left(\vec{\bf{z}}, {\bf{M}}\right)+  \mu\cdot h\left(\bm{\pi}\left(\vec{\bf{z}},{\bf{M}}\right)\right),\]
where $\bm{\pi}\left(\vec{\bf{z}},{\bf{M}}\right)$ is the probability vector formed under $\left(\vec{\bf{z}},{\bf{M}}\right)$ and is clearly continuous in $\vec{\bf{z}}$. The remainder of the proof is the same as that of Theorem \ref{thm_exist} and is therefore omitted. 
\end{proof}

\subsection{Proofs of Section \ref{sec_multidev}}\label{sec_proof_multidev}
In this appendix, write ${\bf{z}}\left(A\right)=\left(z_a\left(A\right)\right)^{\top}_{a \in \mathcal{D}}$ for any set $A \in \Sigma$ of positive measure, as well as  any $N$-partitional contract $\langle \mathcal{P}, w\left(\cdot \right)\rangle$ as the corresponding tuple $\langle A_n, \pi_n, {\bf{z}}_n, w_n \rangle_{n=1}^N$, where $A_n$ is a generic cell of $\mathcal{P}$,  $\pi_n = P_{a^*}\left(A_n\right)$, ${\bf{z}}_n={\bf{z}}\left(A_n\right)$ and $w_n=w\left(A_n\right)$. Assume w.l.o.g. that $w_1  \leq \cdots \leq w_N$. 

\subsubsection{Useful Lemma}
The next lemma generalizes Lemmas \ref{lem_foc} and \ref{lem_mlrp} to encompass multiple agents:

\begin{lem}\label{lem_multidev}
Assume Assumption \ref{assm_mc}. Then for any optimal incentive contract that induces $a^*$, (i) there exists $\bm{\lambda} \in \mathbb{R}_+^{|\mathcal{D}|}$ with $\|\bm{\lambda}\|>0$ such that $u'\left(w_{n}\right)=1/\left(\bm{\lambda}^{\top} {\bf{z}}_n\right)
$ if and only if $w_n>0$; (ii) $\bm{\lambda}^{\top} {\bf{z}}_{1}<0<\bm{\lambda}^{\top} {\bf{z}}_{2}<\cdots$ and  $0=w_1<w_2<\cdots$.  
\end{lem}

\begin{proof}
The wage-minimization problem for given monitoring technology $\langle A_n, \pi_n, {\bf{z}}_n \rangle_{n=1}^N$ is 
\begin{align*}
\min_{\langle \tilde{w}_n \rangle} \sum_{n} \pi_n \tilde{w}_n -\sum_{n} \pi_n u\left(\tilde{w}_n\right)\cdot \bm{\lambda}^{\top} {\bf{z}}_n-\sum_{n} \eta_n \tilde{w}_n, 
\end{align*}
where $\bm{\lambda}$ denotes the profile of the Lagrange multipliers associated with the (\ref{eqn_ic_multidev}) constraints and $\eta_n$ the Lagrange multiplier associated with the (\ref{eqn_ll}) constraint at $\tilde{w}_n$. Note  that $\|\bm{\lambda}\|>0$, because otherwise all (\ref{eqn_ic_multidev}) constraints are slack and hence subtracting a small $\epsilon>0$ from all positive wages constitutes an improvement. Differentiating the objective function with respect to $\tilde{w}_n$ yields the first-order condition in Part (i). The proof of Part (ii) is the same as that of Lemma \ref{lem_mlrp} and is therefore omitted. 
\end{proof}

\subsubsection{Proof of Theorem \ref{thm_multidev}}
\begin{proof}
Take any optimal incentive contract that induces $a^*$. Let $\langle A_n, \pi_n, {\bf{z}}_n, w_n \rangle_{n=1}^N$ be the corresponding tuple and $\bm{\lambda}$ be the profile of the Lagrange multipliers associated with the  (\ref{eqn_ic_multidev}) constraints. Suppose, to the contrary, that some $A_j$ is not $Z_{\bm{\lambda}}$-convex. Then there exist $A', A'' \subset A_j$ and $\tilde{A} \subset A_k$, $k \neq j$ such that (i) $P_{a^*}\left(A'\right), P_{a^*}\left(A''\right), P_{a^*}(\tilde{A})>0$, and (ii) $\bm{\lambda}^{\top} \tilde{\bf{z}}=\left(1-s\right) \bm{\lambda}^{\top} {\bf{z}}'+s \bm{\lambda}^{\top}{\bf{z}}''$, where ${\bf{z}}' \coloneqq {\bf{z}}\left(A'\right)$, $ {\bf{z}}''\coloneqq {\bf{z}}\left(A''\right)$, $\tilde{\bf{z}}\coloneqq {\bf{z}}(\tilde{A})$, $\bm{\lambda}^{\top} {\bf{z}}' \neq \bm{\lambda}^{\top} {\bf{z}}''$ and $s \in \left(0,1\right)$. By Lemma \ref{lem_smallSet}, for all $\epsilon \in (0, \min \{P_{a^*}\left(A'\right), P_{a^*}\left(A''\right), P_{a^*}(\tilde{A})\})$, there exist $A'_{\epsilon}\subset A'$, $A''_{\epsilon} \subset A''$ and $\tilde{A}_{\epsilon} \subset \tilde{A}$ such that (i) $P_{a^*}\left(A'_\epsilon\right)=P_{a^*}\left(A''_\epsilon\right)=P_{a^*}(\tilde{A}_\epsilon)=\epsilon$, and (ii) $\bm{\lambda}^{\top}{\bf{z}}\left(A'_\epsilon\right)=\bm{\lambda}^{\top}{\bf{z}}'$, $\bm{\lambda}^{\top}{\bf{z}}\left(A''_\epsilon\right)= \bm{\lambda}^{\top}{\bf{z}}''$ and $\bm{\lambda}^{\top} {\bf{z}}(\tilde{A}_\epsilon)=\bm{\lambda}^{\top} \tilde{\bf{z}}$.

Consider two perturbations to the monitoring technology: (a) move $A'_\epsilon$ to $A_k$ and $\tilde{A}_\epsilon$ to $A_j$, and (b) move $\tilde{A}_\epsilon$ to $A_j$ and $A''_\epsilon$ to $A_k$. By Assumption \ref{assm_mc}, neither perturbation affects the probability distribution of the output signal under action $a^*$ and hence the monitoring cost. Below we demonstrate that one of them strictly reduces the incentive cost compared to the original (optimal) contract. 

\paragraph{Perturbation (a)} Let $\langle A_n(\epsilon), \pi_n, {\bf{z}}_n\left(\epsilon\right)\rangle_{n=1}^N$ be the tuple associated with the monitoring technology after perturbation (a), where $A_j\left(\epsilon\right)=(A_j\cup \tilde{A}_{\epsilon}) \setminus A'_{\epsilon}$, $A_k\left(\epsilon\right)=\left(A_k \cup A'_{\epsilon}\right) \setminus \tilde{A}_{\epsilon}$ and $A_n\left(\epsilon\right)=A_n$ for $n \neq j, k$. Straightforward algebra shows that
\begin{align}\label{eqn_multitask_zprime}
\begin{cases}
{\bf{z}}_{j}\left(\epsilon\right)={\bf{z}}_j+\displaystyle \frac{{\bf{z}}(\tilde{A}_{\epsilon})-{\bf{z}}\left(A'_{\epsilon}\right)}{\pi_j} \epsilon,\\
{\bf{z}}_{k}\left(\epsilon\right)={\bf{z}}_{k}-\displaystyle \frac{{\bf{z}}(\tilde{A}_{\epsilon})-{\bf{z}}\left(A'_{\epsilon}\right)}{\pi_k}\epsilon, \\
{\bf{z}}_n\left(\epsilon\right)={\bf{z}}_n\text{ } \forall n \neq j, k,
\end{cases}
\end{align}
and that \begin{equation}\label{eqn_zbound_multidev}
\|{\bf{z}}(\tilde{A}_{\epsilon})-{\bf{z}}\left(A'_{\epsilon}\right)\|\leq \|{\bf{z}}(\tilde{A}_{\epsilon})\|+\|{\bf{z}}\left(A'_{\epsilon}\right)\|\leq 2\max_{\omega \in \Omega} \|{\bf{Z}}\left(\omega\right)\|.
\end{equation}
Consider wage profile $\langle w_n\left(\epsilon\right)\rangle_{n=1}^N$ such that (1) $w_1\left(\epsilon\right)=w_1=0$ and (2) all (\ref{eqn_ic_multidev}) constraints are slack by $\mathcal{O}\left(\epsilon\right)$ after the perturbation, i.e., 
\begin{equation}\label{eqn_ic_perturb3}
0\leq \sum_{n=1}^N \pi_n u\left(w_n\left(\epsilon\right)\right)z_{a,n}\left(\epsilon\right) -\sum_{n=1}^N \pi_n u\left(w_n\right)z_{a,n}\sim \mathcal{O}\left(\epsilon\right) \text{ } \forall a \in \mathcal{D}.
\end{equation}
A close inspection of Equations (\ref{eqn_multitask_zprime})-(\ref{eqn_ic_perturb3}) reveals the existence of $M>0$ such that when $\epsilon$ is sufficiently small, there exist wage profiles as above that satisfy $|w_n\left(\epsilon\right)-w_n|<M \epsilon$ for all $n$ and hence the (\ref{eqn_ll}) constraint.\footnote{To see why, define $\kappa_a=\sum_{n=2}^N \pi_n u\left(w_n\right)z_{a,n}$ and $\mathcal{S}_a=\left\{\langle x_n\rangle_{n=2}^N \in \mathbb{R}^{N-1}:  \sum_{n=2}^N x_n z_{a,n} \geq \kappa_a\right\}$ for each $a \in \mathcal{D}$, and note that $\langle \pi_n u\left(w_n\right)\rangle_{n=2}^N \in \displaystyle \cap_{a \in \mathcal{D}} \mathcal{S}_a$. If we cannot construct a wage profile as above, then there exist $a', a'' \in \mathcal{D}$ such that 
$\cap_{a=a', a''}\left\{\langle x_n\rangle_{n=2}^{N}\in \mathbb{R}^{N-1}:  \sum_{n=2}^N x_n z_{a,n} \geq  \kappa_{a}\right\}=\left\{\langle x_n\rangle_{n=2}^{N}\in \mathbb{R}^{N-1}:  \sum_{n=2}^N x_n z_{a',n} = \kappa_{a'}\right\}$ and hence  $z_{a'',n}=-z_{a',n}$ for $n=2,\cdots, N$ and $\kappa_{a''}=-\kappa_{a'}$. In the meantime, $\kappa_{a} \geq c\left(a^*\right)-c\left(a\right)>0$ for all $a \in \mathcal{D}$, thus reaching a contradiction. }

Write $\dot{w}_n= (w_n(\epsilon)-w_n)/\epsilon$ and $\dot{\bf{z}}_n\left(\epsilon\right)=\left({\bf{z}}_n\left(\epsilon\right)-{\bf{z}}_n\right)/\epsilon$. When $\epsilon$ is small, expanding Equation (\ref{eqn_ic_perturb3}) using the twice-differentiability of $u\left(\cdot\right)$ and $|w_n\left(\epsilon\right)-w_n|\sim \mathcal{O}\left(\epsilon\right)$ and multiplying the result by $\bm{\lambda}$ yields 
\[
\sum_{n=1}^N \pi_n \cdot u'\left(w_n\right) \cdot \bm{\lambda}^{\top} {\bf{z}}_n \cdot \dot {w}_n\left(\epsilon\right) = -\sum_{n=1}^N u\left(w_n\right) \cdot \pi_n \cdot \bm{\lambda}^{\top}\dot{\bf{z}}_n\left(\epsilon\right) +\mathcal{O}\left(\epsilon\right). 
\]
Simplifying using $\dot{w}_1\left(\epsilon\right)=0$, $u'\left(w_n\right)=1/\left(\bm{\lambda}^{\top}{\bf{z}}_n\right)$ for $n \geq 2$ (Lemma \ref{lem_multidev}) and Equation (\ref{eqn_multitask_zprime}) yields
\begin{equation}\label{eqn5}
\sum_{n=1}^N \pi_n \dot{w}_n\left(\epsilon\right) = s \left(u\left(w_k\right)-u\left(w_j\right)\right) \left(\bm{\lambda}^{\top}{\bf{z}}''-\bm{\lambda}^{\top} {\bf{z}}'\right) + \mathcal{O}\left(\epsilon\right). 
\end{equation}

\paragraph{Perturbation (b)} Repeating the above argument for perturbation (b) yields
\begin{equation}\label{eqn6}
\sum_{n=1}^N \pi_n \dot{w}_n\left(\epsilon\right) = - \left(1-s\right)  \left(u\left(w_k\right)-u\left(w_j\right)\right) \left(\bm{\lambda}^{\top} {\bf{z}}''-\bm{\lambda}^{\top} {\bf{z}}'\right)+ \mathcal{O}\left(\epsilon\right).
\end{equation}
Since $u\left(w_k\right)\neq u\left(w_j\right)$ by Lemma \ref{lem_multidev} and $ \bm{\lambda}^{\top} {\bf{z}}'' \neq \bm{\lambda}^{\top} {\bf{z}}'$ by assumption, the right-hand sides of Equations (\ref{eqn5}) and (\ref{eqn6}) have the opposite signs when $\epsilon$ is small. The remainder of the proof is the same as that of Theorem \ref{thm_main} and is therefore omitted.
\end{proof}

\subsubsection{Proof of Theorem \ref{thm_exist_multidev}}
\begin{proof}
Define
\[
\Lambda=\left\{{\bm{\lambda}}: {\bm{\lambda}} \in \mathbb{R}_{+}^{|\mathcal{D}|} \text{ and } \left\| \bm{\lambda} \right\|_{|\mathcal{D}|}=1\right\}, 
\]
where $\left\| \cdot \right\|_{|\mathcal{D}|}$ denotes the $|\mathcal{D}|$-dimensional Euclidean norm.
By Theorem \ref{thm_multidev}, any optimal monitoring technology with at most $N \in \left\{2,\cdots, K\right\}$ performance categories is fully captured by $\bm{\lambda} \in \Lambda$ and $N-1$ cutpoints $\widehat{z}_1, \cdots,  \widehat{z}_{N-1}$ such that $\min_{\omega \in \Omega} \bm{\lambda}^{\top} {\bf{Z}}\left(\omega\right) \leq \widehat{z}_1\leq \cdots \leq \widehat{z}_{N-1} \leq \max_{\omega \in \Omega} \bm{\lambda}^{\top}{\bf{Z}}\left(\omega\right)$. Write $\widehat{\bf{z}}=\left(\widehat{z}_1,\cdots, \widehat{z}_{N-1}\right)$. Define  
\[
\mathcal{Z}_{N}\left(\bm{\lambda}\right)=\left\{\widehat{\bf{z}}: \min_{\omega \in \Omega} \bm{\lambda}^{\top} {\bf{Z}}\left(\omega\right) \leq \widehat{z}_1\leq \cdots \leq \widehat{z}_{N-1} \leq \max_{\omega \in \Omega} \bm{\lambda}^{\top} {\bf{Z}}\left(\omega\right) \right\},
\]
equip $\mathcal{Z}_{N}(\bm{\lambda})$ with the sup norm $\|\cdot\|$, and note that $\mathcal{Z}_{N}(\bm{\lambda})$ is compact by Assumption \ref{assm_compact}. For any given pair $\left(\bm{\lambda}, \widehat{\bf{z}}\right)$, write the minimal incentive cost for inducing $a^*$ as $W\left(\bm{\lambda}, \widehat{\bf{z}}\right)$, and note that $W\left(\bm{\lambda}, \widehat{\bf{z}}\right)$ exists and is finite if and only if $\lambda_a>0$ for all $a \in \mathcal{D}$ and $\min_{\omega \in \Omega} \bm{\lambda}^{\top} {\bf{Z}}\left(\omega\right)<\widehat{z}_n<\max_{\omega \in \Omega} \bm{\lambda}^{\top} {\bf{Z}}\left(\omega\right)$ for some $n$. The first condition is necessary: otherwise there exists $a \in \mathcal{D}$ such that $z_a\left(A\right)\equiv 0$ across all performance category $A$'s formed under $\left(\bm{\lambda}, \widehat{\bf{z}}\right)$ and hence the (\ref{eqn_ic_multidev}) constraint will be violated. 

We proceed in two steps.

\paragraph{Step 1} Show that $W\left(\bm{\lambda}, \widehat{\bf{z}}\right)$ is continuous in $\left(\bm{\lambda}, \widehat{\bf{z}}\right)$ for any given $N \in \left\{2,\cdots, K\right\}$. 

\bigskip

Fix any $\bm{\lambda} \in \Lambda$ and $\widehat{\bf{z}} \in \mathcal{Z}_{N}\left(\bm{\lambda}\right)$ such that $W\left(\bm{\lambda}, \widehat{\bf{z}}\right)$ is finite. W.l.o.g. consider the case where $\widehat{z}_n$'s are all  distinct. For sufficiently small $\delta>0$, let $\bm{\lambda}^{\delta}$ and ${\widehat{\bf{z}}}^{\delta}$ be any element of $\Lambda$ and $ \mathcal{Z}_{N}\left(\bm{\lambda}^{\delta}\right)$, respectively, such that $\|\bm{\lambda}^{\delta}-\bm{\lambda}\|_{|\mathcal{D}|}$, $\|{\widehat{\bf{z}}}^{\delta}-\widehat{\bf{z}}\|<\delta$. Let $\pi_n$ and ${\bf{z}}_n$ (resp. $\pi_n^{\delta}$ and ${\bf{z}}_n^{\delta}$) denote the probability (under $a=a^*$) and $|\mathcal{D}|$-vector of $z$-values associated with performance category $A_n=\left\{\omega: \bm{\lambda}^{\top} {\bf{Z}}\left(\omega\right) \in [\widehat{z}_{n-1}, \widehat{z}_n) \right\}$ (resp. $A_n^{\delta}=\left\{\omega: \bm{\lambda} ^{\delta \top} {\bf{Z}}\left(\omega\right) \in [\widehat{z}_{n-1}^{\delta}, \widehat{z}_n^{\delta}) \right\}$), respectively. Let $w_n$ denote the optimal wage at $A_n$. 

Fix any $\epsilon>0$, and consider the wage profile that pays $w_n+\epsilon$ at $A_n^{\delta}$ if $z^{\delta}_{a,n} > 0$ for all $a \in \mathcal{D}$ and $w_n$ otherwise. By construction, this wage profile satisfies the (\ref{eqn_ll}) constraint. Under Assumptions \ref{assm_regular} and \ref{assm_compact}, it satisfies every (\ref{eqn_ic_multidev}) constraint when $\delta$ is small: 
\begin{align*}
&\lim_{\delta \rightarrow 0}\sum_{n} u\left(w_n+ \prod_{a' \in \mathcal{D}}1_{z_{a',n}^{\delta}> 0} \cdot \epsilon \right)\pi_n^{\delta} z_{a,n}^{\delta}\\
&= \sum_{n} u\left(w_n+\prod_{a' \in \mathcal{D}}1_{z_{a',n}> 0} \cdot \epsilon\right)\pi_n z_{a, n}\\
&>\sum_{n} u\left(w_n\right)\pi_n z_{a, n},
\end{align*}
where the inequality holds because that $\sum_{n} \pi_n z_{a',n}=0$ and $z_{a',n}$ is strictly increasing in $n$ for all $a' \in \mathcal{D}$ so there exists $n$ such that $\prod_{a' \in \mathcal{D}}1_{z_{a',n}> 0}=1$. 
To complete the proof, note that 
\[
\lim_{\delta \rightarrow 0}\sum_{n} \pi_n^{\delta}\left(w_n+ \prod_{a \in \mathcal{D}}1_{z_{a,n}^{\delta}> 0} \cdot \epsilon \right)=\sum_{n} \pi_n\left(w_n+\prod_{a \in \mathcal{D}}1_{z_{a,n}>0} \cdot \epsilon\right), 
\]
so the following holds when $\delta$ is sufficiently small: 
\[
W\left(\bm{\lambda}^{\delta}, \widehat{\bf{z}}^{\delta}\right)-W\left(\bm{\lambda}, \widehat{\bf{z}}\right)\leq \sum_{n} \pi_n^{\delta}\left(w_n+\prod_{a \in \mathcal{D}}1_{z_{a,n}^{\delta}> 0} \cdot \epsilon\right)-\sum_{n} \pi_nw_n <\epsilon.
\]
Finally, interchanging the roles between $\left(\bm{\lambda}, \widehat{\bf{z}}\right)$ and $\left(\bm{\lambda}^{\delta}, \widehat{\bf{z}}^{\delta}\right)$ in the above derivation yields $W\left(\bm{\lambda}, \widehat{\bf{z}}\right)-W\left(\bm{\lambda}^{\delta}, \widehat{\bf{z}}^{\delta}\right)<\epsilon$, implying that $\left\vert W\left(\bm{\lambda}^{\delta}, \widehat{\bf{z}}^{\delta}\right)-W\left(\bm{\lambda}, \widehat{\bf{z}}\right)\right\vert<\epsilon$ when $\delta$ is sufficiently small.

\paragraph{Step 2} Under Assumption \ref{assm_cf}(a), the following quantity: 
\[W_N\coloneqq \min_{\bm{\lambda} \in \Lambda, \widehat{\bf{z}} \in \mathcal{Z}_N(\bm{\lambda})}  W\left(\bm{\lambda}, \widehat{\bf{z}}\right)\]
exists and is finite for all $N\in \left\{2,\cdots, K\right\}$ by Step 1 and the compactness of $\Lambda$ and $\mathcal{Z}_{N}\left({\bm{\lambda}}\right)$. Under Assumption \ref{assm_cf}(b), the principal's problem can be written as follows: 
\begin{equation*}
\min_{\bm{\lambda} \in \Lambda, \hat{\bf{z}} \in \mathcal{Z}_K\left(\bm{\lambda}\right)} W\left(\bm{\lambda}, \widehat{\bf{z}}\right) + \mu \cdot h\left(\bm{\pi} \left(\bm{\lambda}, \widehat{\bf{z}}\right)\right), 
\end{equation*}
where $\bm{\pi} \left(\bm{\lambda}, \widehat{\bf{z}}\right)$ denotes the probability vector formed under $\left(\bm{\lambda}, \widehat{\bf{z}}\right)$  and is continuous in its argument. The remainder of the proof is the same as that of Theorem \ref{thm_exist} and is therefore omitted. 
\end{proof}

\section{Other Extensions}\label{sec_online}
\subsection{Individual Rationality}\label{sec_ir}
In this appendix, let everything be as in the baseline model except that the agent is constrained by individual rationality rather than limited liability: 
\begin{equation}\label{eqn_ir}
\tag{IR} \sum_{A \in \mathcal{P}} P_1(A)u(w(A))  \geq c+\underline{u}. 
\end{equation}
A wage scheme is $w:\mathcal{P} \rightarrow \mathbb{R}$, and an optimal incentive contract that induces high effort from the agent (optimal incentive contract for short) minimizes the total implementation cost, subject to the (\ref{eqn_ic}) and (\ref{eqn_ir}) constraints. 
\begin{cor}\label{cor_ir}
Under Assumption \ref{assm_mc}, any optimal monitoring technology that induces high effort from the agent comprises $Z$-convex cells. 
\end{cor}

\begin{proof}
Take any optimal incentive contract and let $\langle A_n, \pi_n, z_n, w_n \rangle_{n=1}^N$ be the corresponding tuple. Assume without loss of generality that $z_1 \leq \cdots \leq z_N$. 

\paragraph{Step 1} Show that $z_1<\cdots<z_N$ and $w_1<\cdots<w_N$.

\bigskip

The wage-minimization problem given $\langle A_n, \pi_n, z_n\rangle_{n=1}^N$ is
\begin{align*}
\min_{\langle \tilde{w}_n \rangle} \sum_{n=1}^N \pi_n \tilde{w}_n - \lambda \left(\sum_{n=1}^N \pi_n u\left(\tilde{w}_n\right) z_n-c\right) -\gamma \left(\sum_{n=1}^N \pi_n u\left(\tilde{w}_n\right) - \left(c+\underline{u}\right)\right), 
\end{align*}
where $\lambda$ and $\gamma$ denote the Lagrange multipliers associated with the (\ref{eqn_ic}) and (\ref{eqn_ir}) constraints, respectively. Differentiating the objective function with respect to $\tilde{w}_n$ and setting the result equal to zero, we obtain  
\[
u'\left(w_n\right)=\frac{1}{\lambda z_n+\gamma}. 
\]
Thus if  $z_j=z_k$ for some $j\neq k$, then $w_j=w_k$. But then merging $A_j$ and $A_k$ has no effect on the incentive cost but strictly reduces the monitoring cost by Assumption \ref{assm_mc}(b), a contradiction to the optimality of the original contract. 

\paragraph{Step 2} Show $Z$-convexity.

\bigskip

Suppose, to the contrary, that some $A_j$ is not $Z$-convex. Consider first perturbation (a) in the proof of Theorem \ref{thm_main}. Take any wage profile $\langle w_n(\epsilon)\rangle_{n=1}^N$ such that the (\ref{eqn_ic}) and (\ref{eqn_ir}) constraints remain binding after the perturbation, i.e.,
\begin{equation}\label{eqn_online_perturbic}
\sum_{n=1}^N \pi_n u\left(w_n\left(\epsilon\right)\right)z_n\left(\epsilon\right)=\sum_{n=1}^N\pi_n u\left(w_n\right)z_n, 
\end{equation}
and 
\begin{equation}\label{eqn_online_perturbir}
\sum_{n=1}^N \pi_n u\left(w_n\left(\epsilon\right)\right)=\sum_{n=1}^N \pi_n u\left(w_n\right).
\end{equation}
A close inspection of Equations (\ref{eqn_zprime}), (\ref{eqn_online_perturbic}) and (\ref{eqn_online_perturbir}) reveals the existence of $M>0$ such that when $\epsilon$ is sufficiently small, there exist wage profiles as above such that $|w_n\left(\epsilon\right)-w_n|<M\epsilon$ for all $n$.\footnote{To see why, define $\kappa_1=\sum_{n=1}^N \pi_n u\left(w_n\right)z_n$, $\kappa_2=\sum_{n=1}^N \pi_n u\left(w_n\right)$, $\mathcal{S}_1=\left\{\langle x_n\rangle_{n=1}^N \in \mathbb{R}^{N}:  \sum_{n=1}^N x_n z_{n} \geq \kappa_1\right\}$ and $\mathcal{S}_2=\left\{\langle x_n\rangle_{n=1}^N \in \mathbb{R}^{N}:  \sum_{n=1}^N x_n \geq \kappa_2\right\}$, and note that $\langle \pi_n u\left(w_n\right)\rangle_{n=1}^N \in \displaystyle \mathcal{S}_1\cap \mathcal{S}_2$. Then from $z_1<\cdots<z_N$, it follows that $\dim \mathcal{S}_1\cap \mathcal{S}_2=N$, and combining with Equation (\ref{eqn_zprime}) gives the desired result.}

Write $\dot{w}_n\left(\epsilon\right)=\left(w_n\left(\epsilon\right)-w_n\right)/\epsilon$ and $\dot{z}_n\left(\epsilon\right)=\left(z_n\left(\epsilon\right)-z_n\right)/\epsilon$, and let $\lambda>0$ and $\gamma>0$ denote the Lagrange multipliers associated with the (\ref{eqn_ic}) and (\ref{eqn_ir}) constraints prior to the perturbation, respectively. Expanding $\lambda$ (\ref{eqn_online_perturbic})+$\gamma$ (\ref{eqn_online_perturbir}) using the twice-differentiability of $u\left(\cdot\right)$ and $|w_n\left(\epsilon\right)-w_n| \sim \mathcal{O}\left(\epsilon\right)$ yields the following when $\epsilon$ is small:  
\begin{equation}\label{eqn_online1}
\sum_{n=1}^N \pi_n \cdot u'\left(w_n\right) \cdot (\lambda z_n+\gamma) \cdot \dot{w}_n\left(\epsilon\right) = - \lambda \sum_{n=1}^N u\left(w_n\right) \cdot \pi_n \dot{z}_n\left(\epsilon\right) +\mathcal{O}\left(\epsilon\right).
\end{equation}
Simplifying using $u'\left(w_n\right)=1/(\lambda z_n+\gamma)$ and Equation (\ref{eqn_zprime}) yields
\begin{equation}\label{eqn_online2}
\sum_{n=1}^N \pi_n \dot{w}_n\left(\epsilon\right)=s \left(u\left(w_k\right)-u\left(w_j\right)\right)  \left(\lambda z''-\lambda z'\right).  
\end{equation}

Consider next perturbation (b). Similar algebraic manipulation as above yields 
\begin{equation}\label{eqn_online3}
\sum_{n=1}^N \pi_n \dot{w}_n\left(\epsilon\right)=-(1-s)  \left(u\left(w_k\right)-u\left(w_j\right)\right)\left(\lambda z''-\lambda z'\right). 
\end{equation}
Since $u\left(w_j\right)\neq u\left(w_k\right)$ and $z'' \neq z'$, we must have $\sgn$ (\ref{eqn_online2}) $\neq \sgn$ (\ref{eqn_online3}), and the remainder of the proof is the same as that of Theorem \ref{thm_main}.

\end{proof}

\subsection{Random Monitoring Technology}\label{sec_random}
This appendix extends the baseline model to encompass random monitoring technologies ${\bf{q}}:\Omega \rightarrow \Delta^K$ mapping raw data points  to elements in the $K$-dimensional simplex. Time evolves as follows: 
\begin{enumerate}
\item the principal commits to $\langle {\bf{q}}, w\rangle$;
\item the agent privately chooses $a \in\{0,1\}$;
\item Nature draws $\omega \in \Omega$ according to $P_a$;
\item the monitoring technology outputs $n \in \{1,\cdots, K\}$ with probability $q_n(\omega)$;
\item the principal pays the promised wage $w_n \geq 0$.
\end{enumerate} 

Under $\langle {\bf{q}}, w\rangle$, the agent is assigned to performance category $n$ with probability 
\[\pi_n=\int q_{n}(\omega) dP_1\left(\omega\right)\]
if he exerts high effort. Define $\mathcal{N}=\left\{n: \pi_n>0\right\}$. For $n \in \mathcal{N}$, define
\begin{equation*}
z_n=\int Z(\omega) q_n(\omega)dP_1(\omega)/\pi_n
\end{equation*}
as the $z$-value of performance category $n$. For $n \notin \mathcal{N}$, let  $w_n=0$.  Then $\langle {\bf{q}}, w\rangle$ is incentive compatible if  
\begin{equation*}
\tag{IC} \sum_{n \in \mathcal{N}} \pi_n u\left(w_n\right) z_n \geq c, 
\end{equation*}
in which case the monitoring cost is proportional to the mutual information of the raw data and output signal conditional on high effort: 
\begin{equation*}
H\left({\bf{q}}, 1\right)=\sum_{n \in \mathcal{N}} \int q_n\left(\omega\right) \log \frac{q_n\left(\omega\right)}{\int q_n\left(\omega\right) dP_1\left(\omega\right)} dP_1\left(\omega\right). 
\end{equation*}
An optimal incentive contract $\langle {\bf{q}}^*, {w}^* \rangle$ that induces high effort from the agent solves
\[
\min_{\langle {\bf{q}}, w \rangle }\sum_{n \in \mathcal{N}} \pi_n w_n+\mu\cdot H({\bf{q}}, 1) \text{ s.t. (\ref{eqn_ic}) and (\ref{eqn_ll})}. 
\]

The next theorem gives characterizations of optimal incentive contracts: 
\begin{thm}\label{thm_mi}
For any optimal incentive contract $\langle {\bf{q}}^{\ast}, w^*\rangle$ that  induces high effort from the agent, we have (i) ${\bf{q}}^*: Z\left(\Omega\right) \rightarrow \Delta^K$; (ii) $\min \left\{w_n^*: n \in \mathcal{N}^*\right\}=0$; (iii) for all $j, k \in \mathcal{N}^{\ast}$, $w_j^* \neq w_k^*$ and $q^*_k\left(z\right)/q^*_j\left(z\right)$ is strictly increasing in $z$ if $w_j^*<w_k^*$.
\end{thm}

\begin{proof}
Since the incentive cost is linear in ${\bf{q}}\left(\omega\right)$ whereas the monitoring cost is convex in ${\bf{q}}\left(\omega\right)$, it follows that ${\bf{q}}^*: Z\left(\Omega\right) \rightarrow \Delta^K$ and that $w_j^* \neq w_k^*$ for all $j, k \in \mathcal{N}^{\ast}$. Write $\mathcal{N}^*=\{1,\cdots, N\}$ and assume w.l.o.g. that $w_1^*<\cdots<w_N^*$. Then $w_1^*=0$ for the same reason as in proof of Lemma \ref{lem_mlrp}. Differentiating the principal's objective function with respect to ${\bf{q}}\left(z\right)$ yields the following first-order condition:
\begin{equation}\label{eqn_foc_random}
-w^*_n+ \lambda u\left(w^*_n\right) z=\mu\left(\log \frac{q_n^*\left(z\right)}{q_1^*\left(
z\right)} - \log\frac{\pi^*_n}{\pi^*_1}\right)  \text{ } \forall n=2,\cdots, N,
\end{equation}
where $\lambda>0$ denotes the Lagrange multiplier associated with the (\ref{eqn_ic}) constraint. The left-hand side of Equation (\ref{eqn_foc_random}) is strictly increasing in $z$, thus proving Part (iii) of this  theorem.
\end{proof}

The next theorem proves existence of optimal incentive contract: 

\begin{thm}
Assume Assumptions \ref{assm_regular} and \ref{assm_compact}. Then an optimal incentive contract that induces high effort from the agent exists. 
\end{thm}

\begin{proof}
For any given $\bf{q}$, the wage-minimization problem admits solutions if and only if $z_j \neq z_k$ for some $j, k \in \mathcal{N}$, in which case we denote the minimal incentive cost by $W\left({\bf{q}}\right)$. The principal's problem is 
\[\min_{{\bf{q}}} W\left({\bf{q}}\right)+\mu\cdot H\left({\bf{q}}, 1\right),\]
and any solution of it must be continuous differentiable on $Z\left(\Omega\right)$ by Equation (\ref{eqn_foc_random}) and Assumptions \ref{assm_regular} and \ref{assm_compact} (taking the usual care of derivatives at end points). Define $C^1\left(Z\left(\Omega\right), \Delta^K\right)$ as the set of ${\bf{q}}$'s as above and equip $C^1\left(Z\left(\Omega\right), \Delta^K\right)$ with the sup norm $\| \cdot\|$, i.e., $\|{\bf{q}}'-{\bf{q}}\|=\sup_{z ,n} |q'_n\left(z\right)-q_{n}\left(z\right)|$. Rewrite the principal's problem as follows: 
\[
\min_{{\bf{q}} \in C^1\left(Z\left(\Omega\right), \Delta^K\right)} W\left({\bf{q}}\right)+\mu \cdot H\left({\bf{q}},1\right),
\]
and note that the objective function is continuous in $\bf{q}$. 

To prove existence of solutions, note that \[\inf_{{\bf{q}} \in C^1\left(Z\left(\Omega\right), \Delta^K\right)} W\left({\bf{q}}\right)+\mu \cdot H\left({\bf{q}},1\right)\] is a finite number, hereafter denoted by $x$. Let $\left\{{\bf{q}}^{k}\right\}$ be any sequence in $C^1\left(Z\left(\Omega\right), \Delta^K\right)$ such that $\lim_{k \rightarrow \infty} W\left({\bf{q}}^k\right)+\mu \cdot H\left({\bf{q}}^k,1\right)=x$. Clearly, ${\bf{q}}^k$ is uniformly bounded for all $k$, and the family $\left\{{\bf{q}}^k\right\}$ is equicontinuous by Assumption \ref{assm_compact} and the definition of $C^1\left(Z\left(\Omega\right), \Delta^K\right)$. Thus, a subsequence of $\left\{{\bf{q}}^k\right\}$ converges uniformly to some ${\bf{q}}^{\infty}$ by Helly's selection theorem, and $W\left({\bf{q}}^{\infty}\right)+\mu \cdot H\left({\bf{q}}^{\infty},1\right)=x$ by the continuity of the objective function. 
\end{proof}

\end{document}